%% file: main.tex
\pdfoutput=1
\documentclass[acmsmall]{acmart}
\usepackage{graphicx} % Required for inserting images
\usepackage{xspace} % For \xspace in macros
\usepackage{amsmath}
\usepackage{makecell}
\usepackage{algorithm}
\usepackage{algpseudocode}
\usepackage{amsthm}
\usepackage{mathpartir}
\usepackage[all]{xy}
\usepackage{multirow}
\usepackage{caption}
\usepackage{graphicx}
\usepackage[most]{tcolorbox}
\tcbuselibrary{skins,breakable}
\usepackage{tikz}
\usepackage{qcircuit}
\usepackage{braket} % \bra, \ket, \braket, \ketbra, \braKet
\usepackage{xcolor}
\usepackage{array,booktabs,tabularx}
\usepackage{wrapfig}
\usepackage{tcolorbox}
\newcolumntype{L}{>{\raggedright\arraybackslash}X}
\usetikzlibrary{calc}

\numberwithin{equation}{section}
\theoremstyle{plain}
\newtheorem{theorem}{Theorem}[section]
\newtheorem{lemma}[theorem]{Lemma}

\setcopyright{cc}
\setcctype{by}
\acmDOI{10.1145/3839450}
\acmYear{2026}
\acmJournal{PACMPL}
\acmVolume{10}
\acmNumber{OOPSLA2}
\acmArticle{318}
\acmMonth{10}
\acmSubmissionID{oopslab26main-p56-p}
\received{2026-03-17}
\received[accepted]{2026-08-06}

\theoremstyle{definition}
\newtheorem{definition}[theorem]{Definition}
\newtheorem{example}[theorem]{Example}

\theoremstyle{remark}
\newtheorem{remark}[theorem]{Remark}
\usepackage{braket}    

\begin{document}
\title{Synthesis of Compact and Expressive Quantum-Circuit Optimizations}

\author{Wei Qiang}
\orcid{0009-0003-2107-1625}
\affiliation{%
  \institution{Columbia University}
  \city{New York}
  \country{USA}
}
\email{wei.qiang@cs.columbia.edu}

\author{Ronghui Gu}
\orcid{0000-0002-6812-6182}
\affiliation{%
  \institution{Columbia University}
  \city{New York}
  \country{USA}
}
\affiliation{%
  \institution{CertiK}
  \city{New York}
  \country{USA}
}
\email{ronghui.gu@columbia.edu}
\include{macro}

\begin{abstract}
Today’s quantum devices are noisy, so reducing circuit size is critical for reliable execution. Existing rule-based optimizers often rely on large rule sets that are difficult to manage and still miss long-distance transformations. We present \sys, a framework for synthesizing compact and expressive quantum-circuit rewrite rules with formal guarantees. We formalize symbolic rewrite rules in which a symbolic gate represents infinitely many subcircuits. We then define canonical symbolic rules of the form $L;S = S;R$ and prove that they constitute a compact generative core from which general symbolic rules can be derived. On top of this formal foundation, given a gate set, \sys synthesizes (1) a small, non-derivable concrete rule set that is complete up to chosen size and qubit bounds, and (2) a small but expressive canonical symbolic rule set that captures transformations beyond finite or monomial-only patterns. We further present rule anchoring to derive optimization-effective rules from canonical symbolic rules. Together, these results provide both expressiveness and guarantees: soundness of synthesized rules via validation, non-derivability, and bounded completeness. On the IBM-Eagle gate set, \sys strictly outperforms state-of-the-art rewrite-based optimizers (Qiskit, Guoq, Quartz, TKET, and Queso) in two-qubit-gate reduction on 90\%, 67\%, 82\%, 85\%, and 83\% of standard quantum algorithm benchmarks, respectively; on Nam gate set, the corresponding rates are  88\%, 74\%, 81\%, 86\%, and 82.9\%. It achieves final average two-qubit-gate reductions of 27.44\% and 29.95\%, respectively.
\end{abstract}

\begin{CCSXML}
<ccs2012>
   <concept>
       <concept_id>10011007.10011006.10011041</concept_id>
       <concept_desc>Software and its engineering~Compilers</concept_desc>
       <concept_significance>500</concept_significance>
       </concept>
   <concept>
       <concept_id>10003752.10003790.10003798</concept_id>
       <concept_desc>Theory of computation~Equational logic and rewriting</concept_desc>
       <concept_significance>300</concept_significance>
       </concept>
   <concept>
       <concept_id>10010583.10010786.10010813.10011726</concept_id>
       <concept_desc>Hardware~Quantum computation</concept_desc>
       <concept_significance>300</concept_significance>
       </concept>
 </ccs2012>
\end{CCSXML}

\ccsdesc[500]{Software and its engineering~Compilers}
\ccsdesc[300]{Hardware~Quantum computation}
\ccsdesc[300]{Theory of computation~Equational logic and rewriting}

\keywords{quantum circuit optimization, rewrite-rule inference, equality saturation, symbolic rewriting}

\maketitle

\section{Introduction}
Quantum computing promises to revolutionize computation by exploiting quantum-mechanical phenomena such as superposition and entanglement to perform tasks that are intractable for classical machines~\cite{tao2025quantum,tao2021gleipnir}. Yet, until recently, quantum hardware has seen limited realization in practice. Today’s quantum hardware—often referred to as Noisy Intermediate-Scale Quantum (NISQ) devices~\cite{preskill2018nisq}—remains severely impacted by short coherence times, gate errors, and restricted qubit connectivity. Recent developments have implemented error-corrected logical qubits and demonstrated the potential to reduce logical errors~\cite{bluvstein2024logical}. As a result, the efficiency of compiled quantum circuits has a direct and often dominant impact on whether an algorithm can be executed reliably on real hardware. Because each operation is noisy, reducing circuit size is essential for lowering error rates. Without optimization, circuits can produce results that are indistinguishable from random noise.

One promising direction is to treat circuit optimization as a process of equational rewriting, where small, equivalent circuit fragments are expressed as rewrite rules. The most prominent approach is peephole optimization, which uses local, bounded-size rewrite rules. Some work relies on domain experts to manually write a small set of rules~\cite{kissinger2020pyzx, VOQC, nam2018automated}, while others synthesize rewrite rules automatically~\cite{queso, quartz}. However, existing techniques still produce a large number of rules, many of which are redundant or derivable from others, and they still fail to capture the full range of useful transformations. For example, previous works synthesize either rewrite rules with symbolic parameters~\cite{quartz} or symbolic rules over monomial gates~\cite{queso}; the latter represent only a finite set of monomial subcircuits, do not induce superposition, and are limited in expressiveness. Both approaches generate large rule sets yet still suffer from limited expressiveness.

The goal of this paper is to answer the following question: Can we automatically synthesize quantum-circuit optimizations that are compact, expressive, and effective in practice?
\sys formalizes symbolic rewrite rules in which symbolic variables represent infinite subcircuit spaces, and their associated constraints characterize the full solution set for which the left-hand side and right-hand side are equivalent. We also formalize canonical symbolic rules of the form $L; S = S; R$, from which many other symbolic rules can be derived. These canonical rules provide a compact generative basis for deriving additional rules. \sys then introduces a novel approach to quantum optimization via rewrite-rule synthesis that addresses both limitations in compactness and expressiveness. We present a framework that, given a gate set, automatically synthesizes: (1) a much smaller set of concrete rewrite rules that are non-derivable yet complete for circuits up to a given size and qubit count, and (2) a small, expressive set of canonical symbolic rewrite rules in which symbolic variables represent infinite subcircuit spaces satisfying specific constraints, including subcircuits that induce superposition.
Although the canonical symbolic rule set is complete up to the $(n,q)$ bounds, many of its rules are size-preserving. To address this, we introduce a novel technique for composing more useful rules from the canonical symbolic rule set, enabling more effective optimization. We demonstrate the effectiveness of our synthesized rules by applying them to optimize quantum circuits using common techniques such as equality saturation~\cite{eqsat} and stochastic search. Our experiments show that our simple optimizer, using the generated rule set, achieves superior optimization results compared to prior approaches that rely on larger, less structured rule sets.
To summarize, \sys contributes the following:
\begin{itemize}
    \item Advances the existing rule inference approach using equality saturation in the synthesis of quantum concrete rules.
    \item Formalizes the notion of symbolic rule and presents a novel synthesis technique that generates symbolic rules efficiently and significantly reduces the term space without compromising completeness.
    \item Introduces a novel rule-composition approach, rule anchoring, that derives more useful rules from the rule set while keeping it small.
    \item Shows that the generated rule set can be integrated with simple optimization algorithms and outperforms state-of-the-art rewrite-based optimizers.
\end{itemize}

\section{Background}
\subsection{Quantum Gates}
A quantum bit (qubit) has two basis states, 0 and 1. A basis state is represented as $\ket{0}$ or $\ket{1}$. A qubit can also be in a linear combination (superposition) of its basis states, $\alpha \ket{0} + \beta\ket{1}$, where $\alpha$ and $\beta$ are complex numbers and $|\alpha|^2 + |\beta|^2 = 1$. $\alpha$ and $\beta$ are called amplitudes. A two-qubit state has four basis states $\ket{00}, \ket{01}, \ket{10}, \ket{11}$ and can be represented as a linear combination of these four basis states. The leftmost qubit is labeled $q_0$, which is the most significant bit, and the rightmost qubit is the least significant bit. Further, an $n$-qubit state can be represented as a linear combination of its $2^n$ basis states. Measurement of a qubit state $\alpha \ket{0} + \beta\ket{1}$ will collapse the qubit to $\ket{0}$ with probability $|\alpha|^2$ and collapse to $\ket{1}$ with probability $|\beta|^2$. A vector representation of a one-qubit state $\alpha \ket{0} + \beta\ket{1}$ is $\begin{bmatrix} \alpha \\ \beta \end{bmatrix}$. A two-qubit state $\alpha \ket{00} + \beta\ket{01} + \gamma\ket{10} + \delta\ket{11}$ can be represented as a vector $\begin{bmatrix} \alpha \\ \beta \\ \gamma \\ \delta \end{bmatrix}$.

Quantum operations transform qubit states of a quantum system. For example, the quantum X gate is analogous to the classical NOT gate: it transforms a qubit state $\alpha \ket{0} + \beta\ket{1}$ to $\beta \ket{0} + \alpha\ket{1}$. The Hadamard gate (H) transforms a qubit state $\ket{0}$ to $\frac{1}{\sqrt{2}}(\ket{0} + \ket{1})$ and $\ket{1}$ to $\frac{1}{\sqrt{2}}(\ket{0} - \ket{1})$. It transforms a basis state to a state in superposition, meaning that measuring the qubit will have a 50\% chance of being 0 and a 50\% chance of being 1. Controlled NOT (CNOT or CX) is a two-qubit gate that flips the second qubit if the first qubit is 1. 

Quantum operations are represented as unitary matrices. For example, the X gate is represented as
$$ X = \begin{bmatrix} 0 & 1 \\ 1 & 0 \end{bmatrix} $$ Applying the X gate to a one-qubit state $\alpha \ket{0} + \beta\ket{1}$ is equivalent to multiplying the matrix X by the state vector $\begin{bmatrix} \alpha \\ \beta \end{bmatrix}$:
$$ X \begin{bmatrix} \alpha \\ \beta \end{bmatrix} = \begin{bmatrix} 0 & 1 \\ 1 & 0 \end{bmatrix} \begin{bmatrix} \alpha \\ \beta \end{bmatrix} = \begin{bmatrix} \beta \\ \alpha \end{bmatrix} $$

Quantum circuits are sequences of quantum gates applied to a set of qubits. 
% A gate set $G$ is $\textbf{universal}$ if any unitary of $2^n \times 2^n$ can be expressed by a finite circuit of $ n$ qubits whose gates are within $G$.
\subsection{Path Sum Representation}
A gate can be written in path-sum notation as
$$\ket{x} \rightarrow \sum_{y \in \mathbb{Z}_{2}} \phi(x,y,\theta) \ket{f(x,y)}$$
Here, $x$ is the input qubit state, $y$ is the path variable that sums over all possible paths, $\theta$ is the symbolic angle parameter, $\phi(x,y,\theta)$ is the phase polynomial that computes the amplitude of each path, and $f(x,y)$ is the output function that computes the output qubit state for each path. 
Intuitively, each path represents a resulting output state. If the gate induces superposition, there are multiple paths. For example, the H gate can be represented as:
$$\ket{x} \rightarrow \frac{1}{\sqrt{2}} \sum_{y \in \mathbb{Z}_{2}} e^{i \pi xy} \ket{y}$$

\noindent If the gate induces no superposition, there is only one path, which we call a monomial gate. For example, the X gate can be represented as: $\ket{x} \rightarrow \ket{\neg x}$.

\subsection{Matrix Semantics}
The semantics of a quantum circuit is a unitary matrix of size $2^n \times 2^n$, where $n$ is the number of qubits. Each entry in the matrix represents the amplitude of transforming one basis state into another. The H gate has the following matrix semantics:
$$ \llbracket H \rrbracket = \frac{1}{\sqrt{2}} \begin{bmatrix} 1 & 1 \\ 1 & -1 \end{bmatrix} $$
Many gates supported by modern quantum devices take symbolic
real-valued parameters; for example, the RZ gate is represented as follows:
$$ \llbracket R_z^{\theta} \rrbracket = \begin{bmatrix} e^{-i\theta/2} & 0 \\ 0 & e^{i\theta/2} \end{bmatrix} $$
% This corresponds exactly to the definition of H gate, which transforms $\ket{0}$ to $\frac{1}{\sqrt{2}}(\ket{0} + \ket{1})$ and $\ket{1}$ to $\frac{1}{\sqrt{2}}(\ket{0} - \ket{1})$. Applying the H gate to $\ket{0}$ is as follows:
% $$H \begin{bmatrix} 1 \\ 0 \end{bmatrix} = \frac{1}{\sqrt{2}} \begin{bmatrix} 1 & 1 \\ 1 & -1 \end{bmatrix} \begin{bmatrix} 1 \\ 0 \end{bmatrix} = \frac{1}{\sqrt{2}} \begin{bmatrix} 1 \\ 1 \end{bmatrix}$$
The semantics of a quantum circuit is computed by matrix multiplication and the tensor product. For example, a circuit with two gates $G_1$ and $G_2$ applied sequentially has the semantics of $M = M_2 \cdot M_1$, where $M_1$ and $M_2$ are the matrix semantics of $G_1$ and $G_2$. A circuit with two gates $G_1$ and $G_2$ applied in parallel on two qubits has the semantics of $M = M_1 \otimes M_2$, where $\otimes$ is the tensor product operator. 
\begin{example}
    Consider a two-qubit circuit in which the H gate is applied to the first qubit and the X gate is applied to the second qubit in parallel, followed by a CNOT gate in which the first qubit is the control and the second is the target. The graph representation of the circuit is
    \[
    \Qcircuit @C=1em @R=.7em {
    & \gate{H} & \ctrl{1} & \qw \\
    & \gate{X} & \targ & \qw
    }
    \]
    Then, the matrix semantics of the circuit is computed as follows: $
        (\llbracket CNOT \rrbracket)(\llbracket H \rrbracket \otimes \llbracket X \rrbracket)
    $.
\end{example}
\subsection{Quantum Rewrite Rule}
\label{sec:rewrite-rule}
Two circuits are equivalent if their matrix semantics are equivalent up to some global phase $\theta$, i.e., $\forall \vec{p}, \exists \theta, M_1(\vec{p}) = e^{i\theta}M_2(\vec{p})$, where $\vec{p}$ are the parameters of the circuits and $M_1,M_2$ are their matrix representations. To eliminate the existential quantifier, we follow an approach similar to that of Quartz~\cite{quartz} and enumerate a list of $\theta$ values that are linear combinations of the symbolic parameters $\vec{p}$ and have been shown to suffice for synthesis. Then, we need only check $M_1(\vec{p}) = M_3(\vec{p})$, where $M_3(\vec{p}) = e^{i\theta} M_2(\vec{p})$.
A rewrite rule is a pair of semantically
equivalent circuits. Rewrite rules can be composed to form a new rule that transforms a longer circuit. Fig.~\ref{fig:rewrite} shows a list of commonly used rewrite rules. Fig.~\ref{fig:rewrite}(a) can be written as $cx\ q_0\ q_1; cx\ q_0\ q_1 \rightarrow empty$.

An Equivalence Circuit Class ($ECC$) is a set of equivalent quantum circuits. $ECCs$ denotes a set of equivalence classes. Any two quantum circuits in an ECC form a valid circuit rewrite. A set of ECCs is $(n, q)$-complete if any valid circuit transformation (rewrite) within size bound $n$ and qubit bound $q$ can be derived by composing transformations from the ECCs. The size of a circuit is its total number of gates, and $q$ denotes the number of qubits used by the circuit.
\begin{figure}[h]
\centering
\begin{minipage}[b]{0.20\textwidth}
    \centering
    \begin{tikzpicture}
        \node[anchor=base] (A) at (0,0) {
    \Qcircuit @C=0.6em @R=1em {
    & \ctrl{1} & \ctrl{1} & \qw\\
    & \targ    & \targ    & \qw
    }
};

\node[anchor=base] (B) at (1.5,0) {
    \Qcircuit @C=0.6em @R=1em {
    & \push{\rule{0pt}{1.4ex}\rule{0.1em}{0pt}} & \qw & \qw \\
    & \push{\rule{0pt}{1.4ex}\rule{0.1em}{0pt}} & \qw & \qw
    }
};
    \draw[->] (A.east) -- ($(B.west |- A.east)$);
    \end{tikzpicture}
\caption*{(a)}
\label{fig:cx-cancel-minimal}
\end{minipage}
\hfill
\begin{minipage}[b]{0.24\textwidth}
    \centering
    \begin{tikzpicture}
        \node (A) at (0,0) {
            \Qcircuit @C=0.4em @R=0.6em {
            & \gate{R_z^\theta} & \ctrl{1} \\
            & \qw & \targ & \\
            }
        };
        \node (B) at (2,0) {
            \Qcircuit @C=0.4em @R=0.6em {
            & \ctrl{1} & \gate{R_z^\theta}\\
            & \targ & \qw \\
            }
        };
        \draw[->] (A) -- (B) node[midway, above] {};
    \end{tikzpicture}
\caption*{(b)}
\label{fig:rz-merge-minimal}
\end{minipage}
\hfill
\begin{minipage}[b]{0.20\textwidth}
    \centering
    \begin{tikzpicture}
        \node (A) at (0,0) {
            \Qcircuit @C=0.6em @R=0.6em {
            & \gate{R_z^{\theta_1}} & \gate{R_z^{\theta_2}} & \qw\\
            }
        };
        \node (B) at (0,-1.2) {
            \Qcircuit @C=0.6em @R=0.6em {
            & \gate{R_z^{\theta_1 + \theta_2}} & \qw\\
            }
        };
        \draw[->] (A) -- (B) node[midway, above] {};
    \end{tikzpicture}
\caption*{(c)}
\label{fig:rz-cx-minimal}
\end{minipage}
\hfill
\begin{minipage}[b]{0.26\textwidth}
    \centering
    \begin{tikzpicture}
        \node (A) at (0,0) {
            \Qcircuit @C=0.6em @R=0.8em {
            & \qw & \ctrl{1} \qw & \qw \\
            & \gate{X} & \targ & \qw \\
            }
        };
        \node (B) at (2,0) {
            \Qcircuit @C=0.6em @R=0.8em {
            & \ctrl{1} & \qw & \qw \\
            & \targ & \gate{X} & \qw \\
            }
        };
        \draw[->] (A) -- (B) node[midway, above] {};
    \end{tikzpicture}
\caption*{(d)}
\label{fig:cx-x-minimal}
\end{minipage}
\caption{Quantum rewrite rule example}
\Description{Four quantum-circuit rewrite rules: two consecutive CX gates cancel; an Rz gate on the control qubit commutes with a CX gate; two consecutive Rz gates merge by adding their rotation angles; and an X gate on the target qubit commutes with a CX gate.}
\label{fig:rewrite}
\vspace{-2em}
\end{figure}
\subsection{Equality Saturation}
An e-graph is a data structure commonly used in theorem provers~\cite{detlefs2005simplify, joshi2002denali, nelson1980techniques} to compactly store equality relations between terms and congruence relations derived from them.
An e-graph consists of a set of e-classes~(equivalence classes), each representing a set of equivalent terms, with each term represented as an e-node. Equality saturation~\cite{eqsat} is an optimization technique that utilizes e-graphs to perform program rewrites. The main idea is to add all possible equivalent terms of a program to the e-graph using a set of rewrite rules and then extract the optimal term according to a cost function. The process of adding equivalent terms to the e-graph is called saturation. During saturation, for each rewrite rule, if the left-hand side (lhs) of the rule matches an e-class in the e-graph, then the right-hand side (rhs) of the rule is added to the same e-class.
Fig.~\ref{fig:egraph} is an example of an e-graph and equality saturation. 
\begin{wrapfigure}{r}{0.4\textwidth}
\centering
\vspace{-1.5em}
\includegraphics[width=0.35\textwidth]{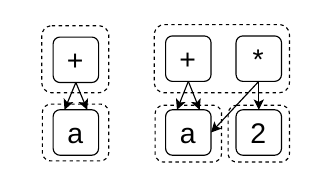}
    \vspace{-1em}
    \caption{The e-graph on the left contains the term $a + a$. Running equality saturation using the rewrite rule $a + a \leftrightarrow a * 2$ yields the e-graph on the right.}
    \Description{Two e-graphs illustrate equality saturation. The left e-graph represents the expression a plus a. The right e-graph adds the equivalent expression a times two to the same equivalence class after applying a rewrite rule.}
    \vspace{-1em}
    \label{fig:egraph}
\end{wrapfigure}
This process is repeated until no more terms can be added to the e-graph. A major advantage of equality saturation engines is that they are non-destructive: they retain all equivalences derived, regardless of the order in which the rules are applied. Finally, the optimal term is extracted from the e-graph by searching for the term with the lowest cost, as defined by a specific cost function.
We utilize equality saturation to check the derivability of a rewrite rule from a set of smaller rewrite rules. If $\mathbb{U}$ is a domain of rewrite rules (e.g., all possible rules within a gate set), a set of rewrite rules is non-derivable if no rule in the set is derivable by composing other smaller rules in $\mathbb{U}$. The size of a quantum-circuit rewrite rule is the maximum number of gates on either its lhs or rhs. Therefore, the rules in Fig.~\ref{fig:rewrite} are non-derivable because no smaller rules exist from which they can be composed.
\begin{figure}
\includegraphics[width=0.8\textwidth, trim=1cm 1cm 1cm 1cm]{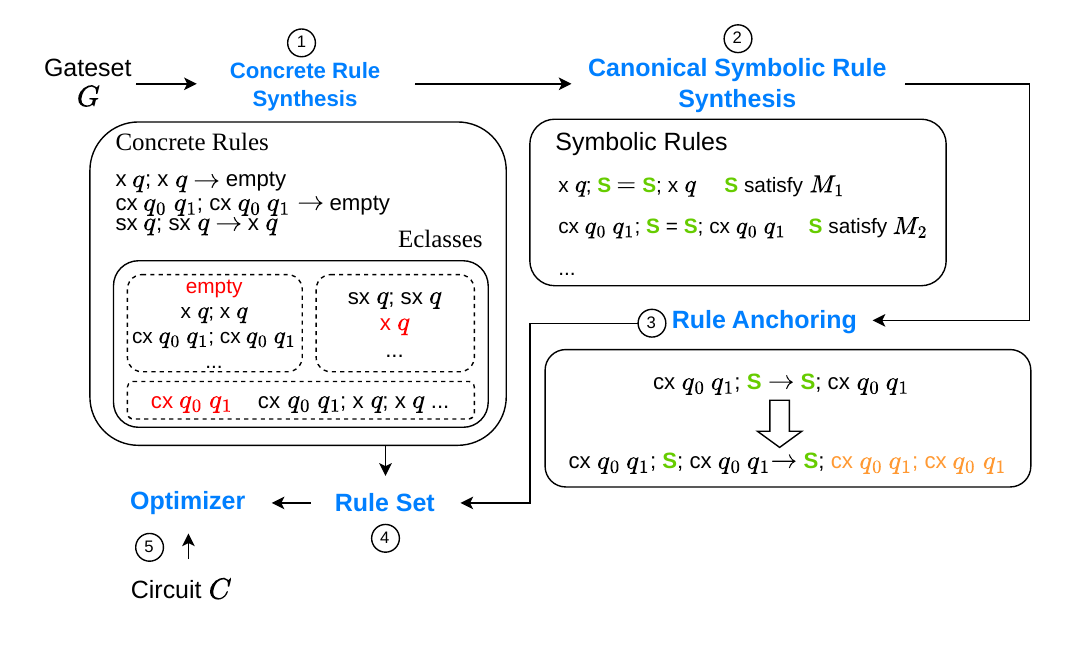}
    \caption{Overview of \sys}
    \Description{The four-stage \sys pipeline. It synthesizes compact concrete rules from grouped circuit equivalence classes, derives canonical symbolic rules with matrix constraints, anchors selected symbolic rules so they enable concrete reductions, and supplies the resulting rule set to a circuit optimizer.}
    \label{fig:overview}
\end{figure}
\section{Overview}
This section describes \sys, a synthesis engine that generates a compact set of symbolic rewrite rules that can represent an infinite space of subcircuits. Each rule has a symbolic subcircuit whose constraints characterize the full set of solutions under which the rule is valid. \sys proves that (1) the rule set is non-derivable, meaning that no rule in the set can be derived from the others, and (2) any symbolic rule whose concrete part lies within the $(n, q)$ bounds can be derived from the canonical rules generated by \sys.

Fig.~\ref{fig:overview} provides an overview of our synthesis pipeline. The synthesis engine takes a gate set, which includes a list of gates and their semantics. In step \textcircled{1}, \sys synthesizes a small set of non-derivable concrete rules that are $(n, q)$-complete for given values of $n$ and $q$ by enumerating all possible circuit terms up to size $n$ and qubit count $q$, grouping them into equivalence classes, and inferring non-derivable rewrite rules from each equivalence class.
% However, although the synthesized concrete rules are minimal and $(n,q)$-complete, composing them into long-distance rewrites remains difficult, even with equality saturation. 
In each e-class (dotted round box), the \textcolor{red}{red} circuit is the class representative, typically the shortest circuit.
In step \textcircled{2}, given the previously synthesized e-classes, \sys synthesizes a set of canonical symbolic rewrite rules. Each e-class representative is enumerated as either the lhs or rhs concrete part of a symbolic rule.
% To reduce the search space, \sys formalizes canonical symbolic rules of the form $L; S = S; R$. We prove that any other symbolic rule can be derived from this form. 
In each symbolic rule, the variable \textcolor{S}{\textbf{$S$}} denotes any subcircuit that satisfies constraint $M_i$, where $M_i$ characterizes the full solution set under which the lhs and rhs are equivalent. In step \textcircled{3}, \sys designs a novel approach that selectively appends prefixes or suffixes to existing canonical symbolic rules so that a concrete rule can be applied afterward. These canonical symbolic rules serve as a generative basis for composing larger and more useful rules that reduce circuit size and collaborate seamlessly with concrete rules. In the figure, the orange part of the anchored rule can match the lhs of a concrete rule, which can reduce that part to an empty circuit. In step \textcircled{4}, concrete rules and anchored symbolic rules are collected into the final rule set, which is used as input to our optimizer.
\begin{figure}[h]
    \centering
    \begin{minipage}[b]{0.25\textwidth}
    \centering
    \begin{tikzpicture}
        \node (A) at (0,1.5) {
            \Qcircuit @C=1em @R=.7em {
            & \gate{R_z^{\theta_1}} & \gate{H} & \qw\\
            %\gategroup{1}{2}{1}{2}{.7em}{--}
            }
        };
        \node (B) at (0,0) {
            \Qcircuit @C=1em @R=.7em {
            & \gate{H} & \gate{R_x^{\theta_1}} & \qw \\
            %\gategroup{1}{2}{1}{2}{.7em}{--}
            }
        };
        \draw[->] (A) -- (B) node[midway, above] {};
        \draw[thick, rounded corners=3pt, draw=green!80!black,
            fill=green!30,
            fill opacity=0.25,
            line width=1pt]
    ($(A.north west)+(1.35cm, 0cm)$) --
    ++(0.8cm,0) -- ++(0,-0.8cm) -- ++(-0.8cm,0) -- cycle;
    \draw[thick, rounded corners=3pt, draw=green!80!black,
            fill=green!30,
            fill opacity=0.25,
            line width=1pt]
    ($(B.north west)+(0.35cm, 0cm)$) --
    ++(0.8cm,0) -- ++(0,-0.8cm) -- ++(-0.8cm,0) -- cycle;
    \end{tikzpicture}
    \caption*{Example (e)}
    \end{minipage}
    \hfill
    % \begin{minipage}[b]{0.3\textwidth}
    % \centering
    % \begin{tikzpicture}
    %     \node (A) at (0,1.5) {
    %         \Qcircuit @C=1em @R=.7em {
    %         \qw & \gate{R_z^{\theta_1}} & \gate{X} & \gate{H} & \qw\\
    %         %\gategroup{1}{3}{1}{4}{.7em}{--}
    %         }
    %     };
    %     \node (B) at (0,0) {
    %         \Qcircuit @C=1em @R=.7em {
    %         \qw & \gate{X} & \gate{H} & \gate{R_x^{\theta_1}} & \qw \\
    %         %\gategroup{1}{2}{1}{3}{.7em}{--}
    %         }
    %     };
    %     \draw[->] (A) -- (B) node[midway, above] {};
    %     \draw[thick, rounded corners=3pt, draw=green!80!black,
    %         fill=green!30,
    %         fill opacity=0.25,
    %         line width=1pt]
    % ($(A.north west)+(1.3cm, 0cm)$) --
    % ++(1.6cm,0) -- ++(0,-0.8cm) -- ++(-1.6cm,0) -- cycle;
    %  \draw[thick, rounded corners=3pt, draw=green!80!black,
    %         fill=green!30,
    %         fill opacity=0.25,
    %         line width=1pt]
    % ($(B.north west)+(0.3cm, 0cm)$) --
    % ++(1.6cm,0) -- ++(0,-0.8cm) -- ++(-1.6cm,0) -- cycle;
    % \end{tikzpicture}
    % \caption*{Example (e)}
    % \end{minipage}
    % \hfill
    \begin{minipage}[b]{0.7\textwidth}
    \centering
    \begin{tikzpicture}
        \node (A) at (0,2) {
            \Qcircuit @C=1em @R=.7em {
            & \gate{R_z^{\theta_1}} & \ctrl{1} & \gate{Z} & \ctrl{1} & \gate{H} & \targ & \gate{X} & \qw\\
            & \qw & \targ & \qw & \targ& \gate{X} & \ctrl{-1} & \qw & \qw
            %\gategroup{1}{3}{1}{4}{.7em}{--}
            }
        };
        \node (B) at (0,0) {
            \Qcircuit @C=1em @R=.7em {
            & \ctrl{1} & \gate{Z} & \ctrl{1} & \gate{H} & \targ & \gate{X} & \gate{R_x^{\theta_1}}\\
            & \targ & \qw & \targ& \gate{X} & \ctrl{-1} & \qw & \qw
            %\gategroup{1}{2}{1}{3}{.7em}{--}
            }
        };
        \draw[->] (A) -- (B) node[midway, above] {};
        \draw[thick, rounded corners=3pt, draw=green!80!black,
            fill=green!30,
            fill opacity=0.25,
            line width=1pt]
    ($(A.north west)+(1.3cm, 0cm)$) --
    ++(4.3cm,0) -- ++(0,-1.5cm) -- ++(-4.3cm,0) -- cycle;
     \draw[thick, rounded corners=3pt, draw=green!80!black,
            fill=green!30,
            fill opacity=0.25,
            line width=1pt]
    ($(B.north west)+(0.3cm, 0cm)$) --
    ++(4.3cm,0) -- ++(0,-1.5cm) -- ++(-4.3cm,0) -- cycle;
    \end{tikzpicture}
    \caption*{Example (f)}
    \end{minipage}
    \caption{Variants of RZ-RX transformation}
    \Description{Two examples of the same long-distance transformation at different scales. An Rz rotation moves across a highlighted intermediate subcircuit and becomes an Rx rotation; the first example uses a single Hadamard gate, while the second uses a longer two-qubit gate sequence.}
    \label{fig:hadamard-rz}
\vspace{-0.8em}
\end{figure}

\noindent\textbf{Step \textcircled{1}. Small $(n,q)$-Complete Concrete Rule Set.}
\sys first constructs a small, non-derivable $(n,q)$-complete rule set that can represent all equivalences up to a given circuit size $n$ and qubit count $q$. It enumerates all circuits of sizes 1 through $n$ on $q$ qubits while removing redundant circuits that are derivable from the rules already learned.
\begin{example}
The figure below shows a transformation $cx\ q0 \ q1; rz(\theta)\ q0; x\ q1; cx\ q0 \ q1$ $\rightarrow$ $rz(\theta)\ q0; x\ q1$. This transformation can be derived using the rules in Fig.~\ref{fig:rewrite} by applying the rules in Fig.~\ref{fig:rewrite}(b), Fig.~\ref{fig:rewrite}(d), and Fig.~\ref{fig:rewrite}(a), in that order. Therefore, this rule is not learned because smaller rules already capture the equivalence.
\end{example}
\makebox[\textwidth][c]{
$\vcenter{\hbox{
    \begin{tikzpicture}
    \node (A) at (0,0) {
        \Qcircuit @C=0.5em @R=.5em {
        & \ctrl{1} & \gate{R_z^{\theta}} & \ctrl{1} & \qw \\
        & \targ & \gate{X} & \targ & \qw \\
        }
    };
    \node (B) at (3,0) {
    \Qcircuit @C=0.5em @R=.5em {
       & \gate{R_z^{\theta}} & \ctrl{1} & \qw & \ctrl{1} \\
        & \qw & \targ & \gate{X} & \targ \\
        }
    };
    \node (C) at (6,0) {
    \Qcircuit @C=0.5em @R=.5em {
        & \gate{R_z^{\theta}} & \ctrl{1} & \ctrl{1} \\
        & \gate{X} & \targ & \targ \\
        }
    };
    \node (D) at (8,0) {
    \Qcircuit @C=0.5em @R=.5em {
        & \gate{R_z^{\theta}} &\qw\\
        & \gate{X} & \qw \\
        }
    };
    \draw[->] (A) -- (B) node[midway, above] {};
    \draw[->] (B) -- (C) node[midway, above] {};
    \draw[->] (C) -- (D) node[midway, above] {};

   \draw[thick, rounded corners=3pt, draw=blue!80!black,
            fill=blue!30,
            fill opacity=0.25,
            line width=1pt]
    ($(A.north west)+(0.2cm, 0cm)$) --
    ++(1.3cm,0) -- ++(0,-0.85cm) -- ++(-0.8cm,0) -- ++(0,-0.6cm) -- ++(-0.5cm,0) -- cycle;
    \draw[thick, rounded corners=3pt, draw=orange!80!black,
            fill=orange!30,
            fill opacity=0.25,
            line width=1pt]
    ($(B.north west)+(1cm, -0.15cm)$) --
    ++(1.1cm,0) -- ++(0,-1.3cm) -- ++(-1.1cm,0) -- cycle;
    \draw[thick, rounded corners=3pt, draw=green!80!black,
            fill=green!30,
            fill opacity=0.25,
            line width=1pt]
    ($(C.north west)+(0.95cm, -0.15cm)$) --
    ++(1cm,0) -- ++(0,-1.3cm) -- ++(-1cm,0) -- cycle;
    \end{tikzpicture}
    }}$
}
\vspace{1em}

Prior work~\cite{quartz,queso} often produces thousands to tens of thousands of rewrite rules, even when restricted to circuits of sizes smaller than six. To address this, given a gate set $G$, we construct a much smaller, non-derivable $(n,q)$-complete rule set that can represent all equivalences up to a given circuit size $n$ and qubit count $q$. The constructed equivalences are then used in symbolic-rule synthesis.

\noindent\textbf{Step \textcircled{2}. Expressive Symbolic Rule Set.}
One limitation of concrete rules is that they are guaranteed to capture only equivalences of concrete circuits enumerated within $(n,q)$. Figure~\ref{fig:hadamard-rz} shows two rewrite rules that share a concrete part but differ in their intermediate parts. For both intermediate parts, RZ can move to the right and become RX. In fact, the space of such intermediate parts is infinite, and there are infinitely many such subcircuits longer than $n$ that concrete rules are not guaranteed to derive. The purpose of an expressive symbolic rule is to represent the full space of such intermediate parts, compactly encoded as a constraint.
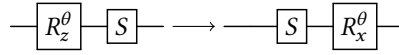
\begin{figure}[h]
    \centering
    \begin{tikzpicture}
        \node (A) at (0,0) {
            \Qcircuit @C=1em @R=.7em {
            & \gate{R_z^{\theta}} & \gate{S} & \qw\\
            % & \qw & \ghost{S} & \qw\\
            }
        };
        \node (B) at (3,0) {
            \Qcircuit @C=1em @R=.7em {
            & \qw & \gate{S} & \gate{R_x^{\theta}} & \qw \\
            % & \qw & \ghost{S} & \gate{R_x^{\theta}} & \qw\\
            }
        };
        \draw[->] (A) -- (B) node[midway, above] {};
    \end{tikzpicture}
    \caption{Circuit graph representation of an RZ-RX transform symbolic rule}
    \Description{A symbolic rewrite in which an Rz rotation before an arbitrary subcircuit S is transformed into an Rx rotation after S, subject to the symbolic rule constraint.}
    \label{fig:symbolic-rzrx}
\end{figure}

\begin{example}
Fig.~\ref{fig:symbolic-rzrx} shows the symbolic rewrite rule that transforms an RZ gate at the start into an RX gate at the end, with a symbolic subcircuit $S$ in between. \sys synthesizes such rules and the constraints on $S$ so that the lhs and rhs are equivalent when the constraints are satisfied. The constraints are represented as a symbolic matrix with variables and are expressive enough to support rules in which $S$ can represent an infinite subcircuit space that may induce superposition. For example, \sys allows the rule in Fig.~\ref{fig:symbolic-rzrx} to match subcircuits that may induce superposition, such as Example (f) in Fig.~\ref{fig:hadamard-rz}.
% such as:
% $$
% rz(\theta)\ q0; cx\ q0\ q1; cx\ q1\ q0; cx\ q0\ q1;  h\ q1; 
% $$
% and transform it to:
% $$
% cx\ q0\ q1; cx\ q1\ q0; cx\ q0\ q1; h\ q1; rx(\theta)\ q1;
% $$ 
\end{example}
\noindent If there is no superposition in $S$, then $S$ is a symbolic monomial gate: modeling the constraint space of $S$ involves only finite permutations in which one qubit state transforms into another because each state is represented by a finite-length bit string. Below is a monomial-gate representation that transforms a qubit state $x$ into another state without superposition.
$$ \ket{x} \rightarrow \ket{f(x)}$$
If $x$ is a one-qubit state, then there are two interpretations of the function: $\{\ket{0} \rightarrow \ket{0}, \ket{1} \rightarrow \ket{1}\}$ and $\{\ket{0} \rightarrow \ket{1}, \ket{1} \rightarrow \ket{0}\}$.
One can simply enumerate all possible finite permutations to model the constraint space and determine the interpretations under which the rule is valid.
However, supporting the synthesis of symbolic rules with a symbolic gate whose solution space includes monomial subcircuits and subcircuits that induce superposition is challenging. Here is the path sum of a circuit with superposition:
$$\ket{x} \rightarrow \sum_{y \in \mathbb{Z}_{2}} \phi(x,y,\theta) \ket{f(x,y)}$$
Every amplitude $\phi$ in every path is a complex number and is subject to the constraint $\sum_{i} |\phi_i|^2 = 1$. Consequently, the path sum admits an infinite number of interpretations. For example, interpretations include $\{\ket{0} \rightarrow \tfrac{1}{\sqrt{2}}(\ket{0}+\ket{1}),\ \ket{1} \rightarrow \tfrac{1}{\sqrt{2}}(\ket{0}-\ket{1})\}$, $\{\ket{0} \rightarrow \tfrac{1}{\sqrt{2}}(\ket{0}+i\ket{1}),\ \ket{1} \rightarrow \tfrac{1}{\sqrt{2}}(i\ket{0}+\ket{1})\}$, $\{\ket{0} \rightarrow \cos\theta\ket{0}+\sin\theta\ket{1},\ \ket{1} \rightarrow -\sin\theta\ket{0}+\cos\theta\ket{1}\}$, etc., where $\theta$ can take infinitely many values.
While prior work~\cite{queso} supports only symbolic rules in which $S$ can match monomial subcircuits, \sys models the solution space of $S$ as a symbolic matrix to compactly represent the full solution space, which consists of an infinite number of interpretations of concrete circuits.

\noindent\textbf{Canonical Rules.} Enumerating all possible symbolic rules of the form $G_i;S;G_j \rightarrow G_i';S;G_j'$ causes combinatorial explosion. However, we prove that every symbolic rule $G_i;S;G_j \rightarrow G_i';S;G_j'$ has a canonical rule represented as $L;S = S;R$, such that $G_i;S;G_j \rightarrow G_i';S;G_j'$ can be derived from its canonical form, and many symbolic rules can be reduced to the same canonical form. This means that, during candidate enumeration, we only need to place the symbolic variable at the two boundaries (prefix/suffix), instead of considering all placements where \(S\) appears in the interior of the circuit, while capturing the same amount of equivalences.
\begin{example}
    \label{ex:canonical-cx}
    The CX cancellation rule (1) has a canonical symbolic rule (2).
    \begin{align}
    cx\ q_0\ q_1; S; cx\ q_0\ q_1 \rightarrow S \tag{1}\\
    cx\ q_0\ q_1; S = S; cx\ q_0\ q_1 \tag{2}
    \end{align}
\end{example}
Intuitively, one can always derive (1) from (2) by right-appending $cx\ q_0\ q_1$ to both sides and applying the concrete rule $cx\ q_0\ q_1; cx\ q_0\ q_1 \rightarrow empty$.

Even with canonical forms, there are still several challenges to address: (1) Enumerating all canonical forms of symbolic rules and checking each candidate individually remains intractable because the search space is still exponential if we naively enumerate $L$ and $R$. (2) Directly solving $S$ for unitary solutions requires solving a list of non-linear equations, which are computationally hard~\cite{calculusofcomputation}. \sys introduces a property-grouping method that groups pairs of $L$ and $R$ with similar properties and proves that, once $L$ and $R$ are in the same group, there exists a solution for $S$ such that $\llbracket S \rrbracket \llbracket L \rrbracket = \llbracket R \rrbracket \llbracket S \rrbracket$. This solves the first challenge by grouping candidates that have solutions for $S$, thereby avoiding checks of all $N^2$ candidates. \sys solves the second challenge by avoiding the direct solution of the unitary constraints on $S$ and deriving only general solutions for $S$ without determining whether a unitary solution exists. However, we show that the property-grouping method can determine whether a unitary solution exists for $S$, allowing us to derive valid rules and check whether a subcircuit satisfies the constraints on $S$ without solving the unitary constraints.

\noindent\textbf{Step \textcircled{3}. Anchoring canonical rules.}
After generating canonical symbolic rules with concrete parts bounded by size $n$ and qubit count $q$, \sys proves that other symbolic rules whose canonical forms are within those bounds can be derived from the canonical rules. Once canonical symbolic rules are generated, they serve as a generative basis that \sys uses to derive more useful rules. However, composing all possible rules is impractical because the set of possible rules is too large. Although the canonical rules already capture the equivalences, they are mostly size-preserving. \sys introduces a novel rule anchoring technique by selectively appending the same prefix or suffix to both sides of a canonical rule. The resulting symbolic rules are designed to collaborate with concrete rules: after a symbolic rule is applied, a concrete rule can be applied to the resulting term to reduce gate size.
\begin{example}
Consider the same rule from Example~\ref{ex:canonical-cx}: $cx\ q_0\ q_1; S = S; cx\ q_0\ q_1$.
\sys first observes that the concrete suffix on the right-hand side, $cx\ q_0\ q_1$, matches the prefix of the lhs of the concrete rule in Fig.~\ref{fig:rewrite}(a): $cx\ q_0\ q_1;cx\ q_0\ q_1 \rightarrow empty$.
Then, \sys "anchors" the canonical rule by appending $cx\ q_0\ q_1$ to the right of both sides:
\[
cx\ q_0\ q_1; S; cx\ q_0\ q_1 \rightarrow S; \tcboxmath[colback=yellow!15,colframe=orange!70!black, boxsep=1pt, left=1pt, right=1pt] {cx\ q_0\ q_1; cx\ q_0\ q_1} \tag{3}
\]
Observe that the circuit in the orange box is now the lhs of the rule $cx\ q_0\ q_1;cx\ q_0\ q_1 \rightarrow empty$, which enables that rule to apply after anchoring.
Intuitively, during synthesis, we retain the anchored rule (3) rather than explicitly deriving rule (1) from the canonical rule. The additional suffix enables rule (3) to match CX gates on both sides of $S$ and move one CX gate next to the other. This creates the orange-boxed subcircuit, to which the concrete cancellation rule in Fig.~\ref{fig:rewrite}(a) can be applied at optimization time to cancel the two cx gates.
\end{example}
\noindent\textbf{Steps \textcircled{4} and \textcircled{5}. Applying the rules.}
After generating both symbolic and concrete rules, we demonstrate that they can be applied using a simple optimizer that combines equality saturation with a stochastic algorithm. Concrete rules are used directly in equality saturation, while symbolic rules are applied stochastically. Our experiments show that, with symbolic rules, this simple optimizer can already outperform most state-of-the-art rewrite-based optimizers.

\noindent\textbf{Theoretical Conclusions.} 
(1) We prove that every symbolic rule has a canonical form and that every symbolic rule can be reduced to and derived from its canonical form. (2) We prove that, for any canonical symbolic rule, if $\llbracket L \rrbracket $ and $\llbracket R \rrbracket$ share the same set of eigenvalues, there exists a unitary solution for $S$ such that $L; S = S; R$. (3) We prove that \sys generates a set of non-derivable canonical symbolic rules.
To keep the main text concise, longer proofs of lemmas and theorems are provided in Appendices A and B.
\section{Synthesizing Concrete Rules}
Previous synthesis engines~\cite{queso, quartz} generate hundreds or thousands of quantum rewrite rules, many of which are derivable from smaller rules. \sys synthesizes a much smaller concrete rule set that can derive larger rules while guaranteeing $(n,q)$-completeness. It also removes rules that are derivable from previously generated rules. \sys combines Ruler-style rule inference~\cite{ruler} with equality saturation to infer a non-derivable rule set, and uses an efficient Polynomial Identity Filter (PIF)~\cite{queso} to group equivalent circuits. Together, these components synthesize concrete rules that represent circuit equivalence classes (ECCs) within $(n,q)$ bounds.
\subsection{Bottom-Up Synthesis Algorithm}
\paragraph{Our Algorithm}
Algorithm~\ref{alg:concrete} shows how \sys synthesizes concrete rules.
\begin{figure}[t]
\centering
\resizebox{0.9\linewidth}{!} {
\begin{minipage}{1.0\textwidth}
\begin{algorithm}[H]
\small
\caption{Concrete Rule Synthesis}
\label{alg:concrete}
\begin{algorithmic}[1]
\Function{SynthesizeConcrete}{n, q}
\State $PIF \gets InitPIFMap()$
\State $L \gets \emptyset$
\For{$i := 1$ to $n$}
\State $X \gets enumerate(i, q)$
\State $PIF.insert(x) $ for each $x$ $\in$ $X$
\State $C \gets \emptyset$
\For{each class EC grouped by PIF}
\State $egraph \gets EmptyEgraph()$
\For{each term $T \in EC$}
\State $egraph.add(T)$
\EndFor
\State $egraph.run\_saturation(L)$
\For{all $(A, B) \in egraph.eclasses \times egraph.eclasses$}
\If{A, B not in same eclass of egraph}
\State C $\gets C \cup \{(A, B)\}$
\EndIf
\EndFor
\EndFor
\State $E \gets \emptyset$ ; 
\State $E \gets Choose(C, L)$
\State $L \gets L \cup Canonicalize(E)$
\EndFor
\State \textbf{return} $L$
\EndFunction
\Statex
\Function{Choose}{C, L}
\State $R \gets \emptyset$
\State $q \gets PriorityQ(C)$  
\While{q is not empty}
\State $r \gets \ q.pop()$
\If{$derivable(r, L \cup R)$}
\State{continue}
\EndIf
\If{$smt.check\_equivalent(r.lhs, r.rhs)$}
\State $R \gets R \cup \{r\}$
\EndIf
\EndWhile
\State \textbf{return} $R$
\EndFunction
\end{algorithmic}
\end{algorithm}
\end{minipage}
}
\Description{Pseudocode for concrete-rule synthesis. Circuits are enumerated by size, grouped using a polynomial identity filter, saturated with previously learned rules in e-graphs, filtered for derivability, validated for equivalence with an SMT solver, and added to the final compact rule set.}
\label{fig:concrete-synthesis}
\end{figure}
Lines 5--6 describe grouping all enumerated terms into classes. PIF~\cite{queso} groups equivalent quantum circuits efficiently with a low failure rate. Using PIF, \sys groups equivalent terms into classes while keeping the probability of mixing semantically different circuits extremely low. The PIF map groups terms by a PIF-computed hash. Within each group, \sys filters out rules that are derivable from the verified rule set $L$ learned in previous iterations, thereby removing duplicates.
In Lines 9--19, the e-graph adds each term of size $i$ in a class and runs saturation with the current learned rules to merge derivable terms. In Lines 15--17, if $A$ and $B$ cannot be proven equivalent under the current learned rules, the tuple $(A,B)$ is added to $C$. Thus, $C$ tracks likely equivalences that are not yet derivable.
Method \texttt{Choose} selects new rules from $C$. It prioritizes candidates using heuristics (e.g., smaller rules first) via a priority queue initialized with $C$. \texttt{Choose} pops candidates in priority order. For each selected rule $r$, it checks whether $r$ is derivable from existing rules $R \cup L$. If it is not derivable, it validates equivalence with the SMT solver. Method \texttt{derivable} checks derivability under $R \cup L$ using equality saturation. In Line 21, newly derived rules $E$ are canonicalized to reduce duplication and then added to the learned rule set $L$.

\noindent\textbf{Correctness.}
Because PIF is used only for grouping, every rule selected by \texttt{Choose} is validated with SMT-based equivalence checking, following Quartz~\cite{quartz}. Therefore, all rules in $L$ are validated, and $L$ always remains sound. In Line 13, the e-graph is saturated only with $L$, preserving soundness during saturation. \sys discards derivable rules early, so only a small set of non-derivable candidates proceeds to SMT verification.

\noindent\textbf{Completeness.}
Given a gate set $G$, function \texttt{enumerate(i, q)} in Line 3 of Algorithm~\ref{alg:concrete} enumerates all circuit terms of size $i$ on $q$ qubits. Therefore, all rewrite rules formed by terms with size up to $n$ and qubit count $q$ can be derived by the rule set synthesized by Algorithm~\ref{alg:concrete}. 

\noindent\textbf{Remove Redundant Rules.}
\texttt{Choose} removes rules that are already derivable by running equality saturation with learned rules $L$ and newly selected rules. If a rule's lhs and rhs are already in the same e-class, that rule is derivable from existing rules. Therefore, \texttt{Choose} selects only rules that are not yet derivable. \sys also adopts pruning techniques from Quartz~\cite{quartz}: (1) \texttt{Enumerate} selects one representative term of size $i$ from each equivalence class to construct terms of size $i+1$; and (2) circuits with common prefix or suffix subcircuits are filtered out. In \texttt{Choose}, the priority queue selects smaller and more abstract rules first. For example, $rz(\theta_1); rz(\theta_2) \rightarrow rz(\theta_1 + \theta_2)$ is selected before $rz(\pi); rz(\pi) \rightarrow rz(\pi + \pi)$; for rules with the same structure, the most general valid rule is selected first.
\noindent\textbf{Rule Canonization.}
The \texttt{Canonicalize} function in Algorithm~\ref{alg:concrete} converts rewrite rules into canonical form to reduce duplication. For example, we generated rules like 
$$x\ q0; x\ q0 \rightarrow empty$$
$$x\ q1; x\ q1 \rightarrow empty$$
The two rules have the same meaning, except they act on different qubits. Thus, all such rules can be canonized into a single rule: $x\ q; x\ q\rightarrow empty$.

\section{Synthesizing Symbolic Rewrite Rules}
Large concrete rule sets are hard to manage and can cause both synthesis and optimization costs to grow rapidly as term size increases. On the other hand, very small rule sets may suffer in optimization quality because long-distance rewrites are difficult to compose. In such cases, either many iterations of equality saturation~\cite{eqsat} are required, or search algorithms such as beam search~\cite{queso, quartz} take a long time to discover large equivalences. For long circuits, equality saturation can quickly fill the e-graph with e-nodes. To address these limitations, we synthesize symbolic rules with symbolic variables, allowing a rule to represent infinitely many subcircuits for which the lhs and rhs are equivalent. These rules can match long-distance subcircuits as long as the constraints are satisfied.
Previous work supports symbolic gate parameters~\cite{quartz, queso} or symbolic monomial gates~\cite{queso}, which represent only a finite set of subcircuit interpretations. However, real circuits often contain superposition, which requires more expressive symbolic rules that can represent subcircuits with superposition.
%  Fig.~\ref{fig:hadamard-rz} is another example showing rule patterns that can represent subcircuits that induce superposition, showing variants of RZ-RX transformation rules where RZ gate can commute gate sequences that are possibly inducing superposition.
\subsection{Symbolic Rule}
A symbolic rule $R_s$ can be represented by the equation $G_1; S; G_2 = G_1'; S; G_2'$, where $G_1, G_2, G_1', G_2'$ are concrete circuit sequences, and $S$ is a symbolic variable representing a set of subcircuits such that, for any subcircuit $C$ in the set, $G_1; C; G_2 = G_1'; C; G_2'$. The $S$ on the lhs is used to match a subcircuit $C$, and the $S$ on the rhs indicates where $C$ should be substituted.
Fig.~\ref{fig:symbolic-rule} shows a circuit-graph representation of the symbolic rule in Example~\ref{ex:canonical-cx}. The goal of \sys is to find a representation of $S$ that captures the complete set of unitary interpretations of $S$ under which the equation holds.
\begin{figure}[h]
\begin{minipage}[b]{0.6\textwidth}
    \centering
    % \begin{tikzpicture}
    %     \node (A) at (0,0) {
    %         \Qcircuit @C=1em @R=.7em {
    %         & \gate{R_z^{\theta_{1}}} & \multigate{1}{S} & \qw & \qw\\
    %         & \qw & \ghost{S} & \gate{R_z^{\theta_{2}}} & \qw\\
    %         }
    %     };
    %     \node (B) at (4,0) {
    %         \Qcircuit @C=1em @R=.7em {
    %         & \qw & \multigate{1}{S} & \qw & \qw \\
    %         & \qw & \ghost{S} & \gate{R_z^{\theta_{1}+\theta_{2}}} & \qw\\
    %         }
    %     };
    %     \draw[->] (A) -- (B) node[midway, above] {};
    % \end{tikzpicture}
    \begin{tikzpicture}
        \node (A) at (0,0) {
            \Qcircuit @C=1em @R=.7em {
            & \ctrl{1} & \multigate{1}{S} & \ctrl{1}& \qw\\
            & \targ & \ghost{S} & \targ & \qw\\
            }
        };
        \node (B) at (3,0) {
            \Qcircuit @C=1em @R=.7em {
            & \qw & \multigate{1}{S} & \qw \\
            & \qw & \ghost{S} & \qw\\
            }
        };
        \draw[->] (A) -- (B) node[midway, above] {};
    \end{tikzpicture}
    \caption{Graph representation of a CX-cancellation symbolic rule.}
    \label{fig:symbolic-rule}
\end{minipage}
\hfill
\begin{minipage}[b]{0.35\textwidth}
\centering
\[ 
S = \begin{bmatrix} a & b & c & c\\ d & m & f & f \\ g & h & j & k \\ g & h & k & j \end{bmatrix}
\]
\caption{Symbolic matrix that represents the intermediate subcircuit of CX cancellation.}
\Description{The left panel shows a symbolic CX-cancellation rule in which an arbitrary two-qubit subcircuit S is enclosed by two CX gates and the rewrite removes both CX gates. The right panel gives the structured four-by-four symbolic matrix that characterizes valid instantiations of S.}
\label{fig:symbolic-matrix}
\end{minipage}
\vspace{-1.4em}
\end{figure}
\subsection{Constraints}
\label{sec:constraint}
\sys represents the solution of $S$ as a symbolic matrix that captures all unitary matrices to which $S$ can be instantiated while keeping the lhs and rhs equivalent. During matching, a symbolic subcircuit can match any subcircuit whose unitary matrix is represented by this symbolic matrix.
\newpage
\begin{wrapfigure}{r}{0.32\textwidth}
    \centering
    \( \begin{cases} |a|^2 + |b|^2 + |c|^2 + |c|^2 = 1 \\
|d|^2 + |m|^2 + |f|^2 + |f|^2 = 1 \\
|g|^2 + |h|^2 + |j|^2 + |k|^2 = 1 \\
a d^* + b m^* + c f^* + c f^* = 0 \\
... \end{cases} 
    \)
    \vspace{-1em}
    \caption{The constraints corresponding to Fig.~\ref{fig:symbolic-matrix}}
    \Description{A representative subset of the nonlinear equations imposed by the unitarity condition on the symbolic matrix S, including row normalization equations and orthogonality equations between rows.}
    \label{fig:symbolic_constraints}
    \vspace{-1em}
\end{wrapfigure}
Fig.~\ref{fig:symbolic-matrix} shows how the intermediate part of a symbolic rule (e.g., CX cancellation) can be represented as a $4 \times 4$ symbolic unitary matrix $S$, where $a, b, c, \ldots, j$ are complex variables satisfying the unitary constraint $SS^{\dagger} = I$. Expanding this constraint yields a set of non-linear equations, illustrated in Fig.~\ref{fig:symbolic_constraints}.
In the figure, $|x|^2$ denotes $x x^*$, where $x^*$ is the complex conjugate of $x$. To determine whether $S$ has valid solutions, \sys must find a symbolic matrix form (as in Fig.~\ref{fig:symbolic-matrix}) that satisfies both constraints: (1) $G_1; S; G_2 = G_1'; S; G_2'$ and (2) $S$ has a unitary solution.
However, enumerating circuits in (1) can cause combinatorial explosion, and solving non-linear equations in (2) is known to be computationally expensive~\cite{calculusofcomputation} and to scale poorly as the matrix dimension increases.
We now discuss how \sys avoids checking the non-linear equations and infers candidate solutions for $S$.

\subsection{Canonical Symbolic Rules and Property Grouping}
Having defined symbolic rules and constraints, we now describe how \sys synthesizes symbolic rules. There are two key challenges. (1) Unlike concrete-rule synthesis, where likely equivalent concrete terms can be grouped efficiently using probabilistic equivalence checking, symbolic terms are hard to group because symbolic variables have infinitely many possible instantiations. (2) If there are $N$ concrete terms, checking whether each candidate admits a valid symbolic constraint requires considering an exponential number of candidates. For example, for two concrete candidate terms $G_1; G_2; G_3$ and $G_4; G_5$, we need to insert a symbolic variable at every position and produce an exponentially large number of symbolic-rule candidates:
$G_1; S; G_2; G_3 = G_4; S; G_5$, $G_1; G_2; S; G_3 = G_4; G_5; S$, etc. Checking every candidate to determine whether such an $S$ exists is intractable.
We now describe how \sys addresses these challenges.

\subsubsection{\textbf{Canonical Form of Symbolic Rules}}
Given a candidate symbolic rule $G_1; S; G_2 = G_1'; S; G_2'$, where $S$ is a symbolic variable representing a subcircuit satisfying the constraint $M$, we describe how \sys derives the symbolic-matrix constraint or discards the candidate if no valid $S$ exists. We show that every candidate of this form can be reduced to, and derived from, a canonical-form representation $L; S = S; R$.
\begin{theorem}
\label{lem:canonical}
Given a gate set $A$, where each gate has a finite sequence of gates in $A$ that implements its inverse, and two circuits with a symbolic gate, $G_1; S; G_2$ and $G_1'; S; G_2'$, where $S$ is a variable that represents subcircuits satisfying the constraint $M$ and each $G_i$ is a concrete circuit, there exist circuits $L$ and $R$ implemented over $A$ such that $\llbracket L \rrbracket = \llbracket G_1 \rrbracket \llbracket G_1'\rrbracket^{-1}$ and $\llbracket R \rrbracket = \llbracket G_2 \rrbracket^{-1} \llbracket G_2' \rrbracket$. $G_1; S; G_2$ and $G_1'; S; G_2'$ are equivalent iff $L; S$ and $S; R$ are equivalent.
\end{theorem}
\begin{proof}
    If each gate has a finite inverse implementation in $A$, there is always a finite circuit sequence that implements the inverse of $G_1$ and $G_2$. We can left-append a circuit that implements $\llbracket G_1 \rrbracket^{-1}$ and right-append a circuit that implements $\llbracket G_2 \rrbracket^{-1}$ to both sides of the equation to obtain the canonical form. The reverse direction is similar.
\end{proof}
We call $(\llbracket L \rrbracket, \llbracket R \rrbracket)$ the canonical form of $G_1; S; G_2 = G_1'; S; G_2'$. A canonical form is valid if the constraint $M$ on $S$ contains at least one solution such that the rule's lhs and rhs are equivalent. The rule $L;S = S;R$ is one circuit-level representation of this canonical form, which we call a canonical symbolic rule. Since each matrix can have many equivalent circuit representations, we discuss below how \sys enumerates only one representation for each canonical form.

The theorem's premise that each gate has a finite inverse implementation in the gate set holds for gate sets that arise in practice, including native gate sets for real hardware such as IBM~\cite{qiskit}, IonQ~\cite{ionq2022native}, Rigetti, and others, as well as standard textbook gate sets such as Nam~\cite{nam2018automated} and Clifford+T~\cite{Gottesman_1998}. This theorem allows \sys to enumerate only canonical symbolic rules, where $S$ appears at the two ends of the rule and on opposite sides (i.e., $L;S$ on the lhs and $S;R$ on the rhs). This significantly reduces the symbolic-rule search space. After generating a canonical symbolic rule represented as $L; S = S; R$ with symbolic variable $S$, we can derive rules of the form $G_1; S; G_2 = G_1'; S; G_2'$ by left-appending subcircuits that implement the unitary matrix $A$ and right-appending subcircuits that implement the unitary matrix $B$ to both sides, such that $G_1 = AL$ and $G_2 = RB$.
\begin{example}
\label{ex:canonical}
The RZ-merging symbolic rule: $rz(\gamma)\ q0; S; rz(\theta)\ q1 = S; rz(\gamma + \theta)\ q1$
has the canonical form:
$$ rz(\gamma)\ q0; S = S; rz(\gamma)\ q1$$
\end{example}
We can derive the RZ-merging symbolic rule by right-appending $rz(\theta)\ q1$ to both sides of the canonical form and then using a simple concrete rule in Fig.~\ref{fig:rewrite}(c) to merge two RZ gates on the right-hand side.
\subsubsection{\textbf{Enumerate Candidates}}
Even though we have reduced the search space of symbolic rules to canonical form, it is still intractable to enumerate all possible candidates for $L,R$. The representation of a canonical form is not unique because circuits $L$ and $R$ can be represented by many equivalent circuits in an equivalence class. \sys further reduces the search space by enumerating one concrete circuit to represent $L$ or $R$ from each equivalence class generated during concrete-rule synthesis and appending a symbolic variable $S$ to it. Thus, for each canonical form, we enumerate only one representation. The intuition is that, while concrete rules already capture the equivalences between concrete circuits, we need to select only one representative from each equivalence class to serve as the concrete part of a canonical symbolic rule.

\subsubsection{\textbf{Group Symbolic Terms}}
For concrete terms, we can efficiently group circuits approximately using PIF. However, for symbolic terms, because of the infinite number of possible instantiations of symbolic gates, it is hard to group symbolic terms into equivalence classes. \sys groups symbolic terms with a variable $S$ using necessary conditions for the existence of a solution. If the concrete components $L$ and $R$ are placed in the same group, there may exist an instantiation of $S$ satisfying $L;S = S;R$; if they are placed in different groups, no such instantiation exists.
\begin{lemma}
\label{lem:eigen}
    Given unitary matrices $L_m, R_m$, there exists a unitary matrix $S$ such that $SL_m = R_mS$ if and only if $L_m$ and $R_m$ have the same eigenvalues, and each eigenvalue appears the same number of times (i.e., has the same multiplicity) in both matrices.
\end{lemma}
\begin{proof}
If there exists a unitary matrix $S$ such that $SL_m = R_mS$, then multiplying both sides on the right by $S^{\dagger}$ gives
$SL_mS^{\dagger} = R_mSS^{\dagger} = R_m$.
Hence, $L_m$ and $R_m$ are unitarily similar matrices, so they have the same set of eigenvalues with the same multiplicities.

Conversely, if $L_m$ and $R_m$ have the same set of eigenvalues with the same multiplicities, then by the spectral theorem, there exist unitary matrices $U$ and $V$ such that
$L_m = UDU^{\dagger}$ and $R_m = VDV^{\dagger}$,
where $D$ is a diagonal matrix whose diagonal entries are the eigenvalues of $L_m$ and $R_m$.
Let $S = VU^{\dagger}$. Then
\[
SL_m = VU^{\dagger}UDU^{\dagger} = VDU^{\dagger} = VDV^{\dagger}VU^{\dagger} = R_mS.
\]
\end{proof}
\begin{lemma}
\label{lem:trace}
    Given unitary matrices $L_m$ and $R_m$, if there exists a unitary matrix $S$ such that $SL_m = R_mS$, then $L_m$ and $R_m$ have the same trace.
\end{lemma}
\begin{proof}
    The trace of a matrix is the sum of its eigenvalues. Thus, if $L$ and $R$ have the same set of eigenvalues with the same multiplicities, they must have the same trace.
\end{proof}
With the above two lemmas, \sys can group symbolic terms into classes, each of which contains symbolic terms that have the same trace and the same set of eigenvalues with the same multiplicities. This significantly reduces the search space of symbolic terms to be considered. Because computing traces is relatively efficient, we first group terms by their traces and then further test each term's concrete eigenvalues within each trace group. Because identical traces are a necessary but insufficient condition, they are used as a filter to eliminate a large number of non-existent symbolic rules early, before computing the eigenvalues of $L$ and $R$. For $L$ and $R$ that have symbolic variables such as angles, we again utilize probabilistic grouping by sampling concrete parameters to eliminate many non-existent symbolic rules whose lhs and rhs do not have equal traces. After that, for each rule in the group, we check whether the eigenvalues are symbolically equivalent with an SMT solver.

\subsubsection{Solving Symbolic Matrix}
Once we have grouped symbolic terms into classes, \sys can enumerate each circuit pair $(L, R)$ and construct solutions to the equation $\llbracket S \rrbracket \llbracket L \rrbracket = \llbracket R \rrbracket \llbracket S \rrbracket$. Solving this equation is equivalent to solving a set of linear equations and finding the nullspace of the matrix
$K(\llbracket R \rrbracket, \llbracket L \rrbracket) \;=\; I_n \otimes \llbracket R \rrbracket \;-\; \llbracket L \rrbracket^{\mathsf T} \otimes I_n$, if $\llbracket L \rrbracket \in \mathbb{C}^{n\times n}$ and $\llbracket R  \rrbracket \in \mathbb{C}^{n\times n}$, where $I_n$ is the identity matrix of size $n$.
\begin{example}
Given a canonical symbolic rule $cx\ q0\ q1; S = S; cx\ q0\ q1,$
\sys tries to solve for $S$ such that $\llbracket S \rrbracket \llbracket CX \rrbracket = \llbracket CX \rrbracket \llbracket S \rrbracket$. Solving the equation gives a list of basis matrices $\{B_1,...,B_k\}$ that correspond to the solution space of $S$. After assigning a coefficient variable $c_1..c_k$ to each basis matrix, the derived symbolic matrix is equivalent, modulo variable renaming, to the symbolic matrix in Fig.~\ref{fig:symbolic-matrix}.
\end{example}
\subsection{Synthesis Algorithm}
Algorithm~\ref{alg:symbolic} shows the algorithm \sys uses to synthesize symbolic rules. Line 2 shows that \sys first collects representative terms from each equivalence class produced by concrete-rule synthesis. ECs are computed by running equality saturation on the final concrete rule set $L$ and extracting a representative from each e-node. Line 3 groups the circuit terms $T$ by their concrete traces, with the symbolic angles in $L,R$ evaluated using $n$ sampled concrete values. This step eliminates many invalid candidates early. Lines 4--14 show that, for each grouped class $C$, \sys enumerates each pair of terms $(L, R)$ in $C$. Line 6 solves $\llbracket S \rrbracket \llbracket L \rrbracket = \llbracket R \rrbracket \llbracket S \rrbracket$ as a linear system. The result is a set of basis matrices whose linear combinations form the solution space of $S$. In Line 8, we validate the eigenvalues of $L$ and $R$ symbolically to determine whether a unitary solution for $S$ exists.
The grouping steps significantly reduce the number of systems of linear equations to be solved and eliminate a large number of non-existent symbolic rules early, avoiding the need to solve $S$'s unitary constraints directly; these constraints form a list of non-linear equations.
\begin{algorithm}[t]
\small
\caption{Symbolic Rule Synthesis}
\label{alg:symbolic}
\begin{algorithmic}[1]
\Function{SynthesizeSymb}{ECs}
\State $T \gets$ \{ Representative(EC) | for each equivalence class EC $\in$ ECs\}
\State $CC \gets$ group T by $(concrete\_trace(T))$ \Comment{CC is a set of equivalence classes}
\For{each class $C \in CC$}
\For{each pair of terms (L, R) in C}
\State $B \gets$ SolveIntertwiner(L, R) \Comment{Solve $\llbracket S \rrbracket \llbracket L \rrbracket$ = $\llbracket R \rrbracket \llbracket S \rrbracket$}
\If {$B \neq \{\}$}
\If {eigenvalues(L) = eigenvalues(R)} \Comment{validate the eigenvalues}
\State $SymbolicRules \gets SymbolicRules \cup \{(L, R, B)\}$
\EndIf
\EndIf
\EndFor
\EndFor
\EndFunction
\end{algorithmic}
\end{algorithm}
\noindent\textbf{Correctness} By Lemmas~\ref{lem:eigen} and~\ref{lem:trace}, for each pair of terms $L,R$, there exists a unitary matrix $S$ such that $S\llbracket L \rrbracket = \llbracket R \rrbracket S$ iff $eigenvalues(L)$ = $eigenvalues(R)$. For $L, R$ that have symbolic parameters such as angles, we first sample concrete angles and compute each circuit's concrete trace and then validate the equivalence of the eigenvalues symbolically in Line 8 using the Z3 SMT solver~\cite{z3} if $B$ is nonempty. Thus, all $B$ derived by solving $L;S = S; R$ have unitary solutions to the equations.

\noindent\textbf{Completeness}
Because we take one representative from each equivalence class generated by Algorithm~\ref{alg:concrete}, and Algorithm~\ref{alg:concrete} generates an $(n,q)$-complete non-derivable concrete rule set, the terms $(L,R)$ are canonicalized to representatives, and Algorithm~\ref{alg:symbolic} generates all canonical symbolic rules within qubit bound $q$ and with concrete parts within bound $n$.
It is not necessary to include all terms in an equivalence class because, once we have a canonical form represented as $L;S = S;R$, applications of concrete rules can always replace $L$ and $R$ with any of their previously generated equivalent terms. More intuitively, the concrete rules have already captured the equivalence relations between concrete terms.
\begin{lemma}
    For any symbolic rule $\phi$ and its canonical form represented as $L; S = S; R$, if $L$ and $R$ are within size $n$ and qubit bound $q$, then it can be derived from the canonical symbolic rules and concrete rules generated by Algorithms~\ref{alg:symbolic} and~\ref{alg:concrete}.
\end{lemma}
\begin{proof}
It follows from Theorem~\ref {lem:canonical} that $\phi$ can be derived from the canonical rule.
\end{proof}
\noindent\textbf{Non-derivability}
We also show that each generated canonical symbolic rule is non-derivable from the other generated canonical symbolic rules.
\begin{definition}{(Symbolic Matrix Satisfiability)} Given a Symbolic Matrix $S$ with variables $v_1, v_2,...v_n \in \mathbb{C}$, we denote $M \models S$ if there exists a model $M$ that assigns complex numbers or complex variables to $v_1, v_2,...v_n$.
\end{definition}
A model $M$ can be viewed as a matrix containing concrete numbers and/or symbolic parameters. The matrix semantics of a quantum circuit is a unitary matrix that may also contain symbolic parameters (e.g., angles), which can be checked against $S$. If it is a model of $S$, then it satisfies the constraints of $S$.
\begin{definition}
    For any two symbolic matrices $S, S'$, we denote $S \subseteq S'$ if $M \models S$ implies $M \models S'$ for any model M.
\end{definition}
\begin{lemma}
\label{lem:uniquer}
Given a unitary matrix $L$ and a matrix $S$ with symbolic parameters, there exists a unique unitary matrix $R$, up to variable renaming, such that $SL = RS$.
\end{lemma}
\begin{remark}
The symmetric version also holds for a unique $L$ such that $SL = RS$, given $S$ and $R$. The lemma implies that the two unitary symbolic matrices $L$ and $S$ determine the $R$ such that $SL = RS$. The lemma is later used to prove the following theorem.
\end{remark}
\begin{theorem}
Any canonical symbolic rule $R_s$ is not derivable from the other generated canonical symbolic rules.
\end{theorem}
% \begin{proof}
% Suppose a symbolic rule $T$, without loss of generality (the symmetrical rule has a similar proof), to be $L;S \rightarrow S;R$ is derivable by some other canonical rules. Then, there exists a sequence of rules $T_1...T_k$ not the same form as $T$ that can rewrite $L;S$ to $S; R$. Then, the lhs of $R_1$ must be of form $L; S_1$, where $S_1 \subseteq S$. The rhs of $R_k$ must be of the form $S_k; R$, where $S_k \subseteq S$. Notice that every $S_{i+1} \subseteq S_i$. This gives $S \subseteq S_k \subseteq S_{k-1}.... \subseteq S$, then all $S_i$ are actually equivalent to S modulo variable renaming. By lemma~\ref{lem:uniquer},
% every $T_i$ is the form $L;S \rightarrow S;R$ or $S;R \rightarrow L;S$, meaning that $R_i$ are just transforming L and R back and forth using the same rule, which contradicts the assumption that $R_i$ are not rule $L;S \rightarrow S; R$.
% \end{proof}
These results imply that we generate a compact symbolic-rule set that represents a large family of circuit equivalences, because symbolic variables can represent infinitely many unitary matrices.
\subsection{Checking Constraints}
Once we have a valid symbolic rule $R$, we match a subcircuit by syntactically matching the lhs pattern of the symbolic rule, producing a substitution that maps the variable $S$ to a concrete subcircuit $C$,
$\delta = \{S : C\}$.
The concrete part of the candidate must match the rule's concrete part, while the symbolic part $S$ can match any subcircuit. After matching, we check whether $\delta(S)$ satisfies the constraint on $S$ by testing whether the unitary matrix $M$ of $C$ is a model of symbolic matrix $S$, i.e., $M \models S$. \sys checks whether each entry of $M$ can be assigned consistently to variables in $S$ (the same variable in $S$ must map to the same value in $M$). If such an assignment exists, the constraint is satisfied. Since $M$ is already unitary (from the matrix semantics of the matched subcircuit), we do not re-check unitarity constraints. \sys avoids directly solving the non-linear equations in \autoref{sec:constraint} by grouping terms for which unitary solutions for $S$ exist during candidate enumeration and by solving linear equations to obtain a general symbolic matrix without explicit unitarity constraints. Intuitively, we do not need to solve for $S$ to check whether $M \models S$ for a given unitary matrix $M$.
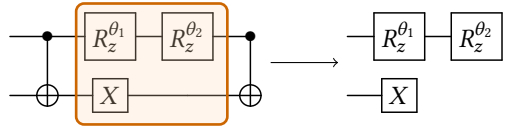
\begin{figure}[h]
\begin{minipage}[b]{0.5\textwidth}
    \centering
    \begin{tikzpicture}
        \node (A) at (0,0) {
            \Qcircuit @C=1em @R=.7em {
            & \ctrl{1} & \gate{R_z^{\theta_1}} & \gate{R_z^{\theta_2}} & \ctrl{1}\\
            & \targ & \gate{X} & \qw & \targ \\
            %\gategroup{1}{3}{2}{4}{.7em}{--}
            }
        };
        \node (B) at (4,0) {
            \Qcircuit @C=1em @R=.7em {
            & \gate{R_z^{\theta_1}} & \gate{R_z^{\theta_2}} &\\
            & \gate{X} &\\
            }
        };
        \draw[->] (A) -- (B) node[midway, above] {};
        \draw[thick, rounded corners=3pt, draw=orange!80!black,
            fill=orange!30,
            fill opacity=0.25,
            line width=1pt]
    ($(A.north west)+(1cm, 0cm)$) --
    ++(2cm,0.0cm) -- ++(0,-1.6cm) -- ++(-2cm,0) -- cycle;
    \end{tikzpicture}
\end{minipage}
\caption{Example of matching a symbolic rule on a CX-framed subcircuit.}
\Description{A circuit containing two outer CX gates and an orange-highlighted intermediate subcircuit is matched by the CX-cancellation symbolic rule. Rewriting removes the two CX gates and retains the highlighted Rz and X gates.}
\label{fig:cx_example}
\end{figure}
\begin{example}
The symbolic rule $$cx\ q0\ q1; S; cx\ q0\ q1 \rightarrow S;$$
can match a circuit in Fig.~\ref{fig:cx_example}; the symbolic part $S$ is matched to the subcircuit in the orange box.
The concrete unitary matrix $M$ is:
$$ M = (I \otimes \llbracket X \rrbracket)(\llbracket RZ(\theta_2) \rrbracket \llbracket RZ(\theta_1)\rrbracket \otimes I)
= \begin{bmatrix} 0 & e^{-i (\theta_1 + \theta_2)/2} & 0 & 0 \\ e^{-i (\theta_1+ \theta_2)/2} & 0 & 0 & 0 \\ 0 & 0 & 0 & e^{i (\theta_1 + \theta_2)/2} \\ 0 & 0 & e^{i (\theta_1 + \theta_2)/2} & 0 \end{bmatrix} $$
It is straightforward to check that there exists an assignment of complex numbers to $S$ in Fig.~\ref{fig:symbolic-matrix} such that $M \models S$:
$$ a = 0, b = e^{-i (\theta_1 + \theta_2)/2}, c = 0, d = e^{-i (\theta_1 + \theta_2)/2}, m = 0, f = 0, g = 0, h = 0, k = e^{i (\theta_1 + \theta_2)/2}, j = 0 $$
Thus, the constraint is satisfied, and we can apply the symbolic rule to optimize the matched subcircuit.
\sys also supports matching a $k$-qubit symbolic matrix $S$ ($k < n$) with an $n$-qubit subcircuit $C$. For example, if $S$ is a $4 \times 4$ matrix that constrains only qubits $x_1,x_2$, and the matched subcircuit contains $n>2$ qubits $x_1,x_2,\ldots,x_n$, then \sys checks that $C$ transforms $x_1,x_2$ correctly regardless of other qubits' states. In addition, \sys checks that the remaining qubits yield consistent results for $L;C$ and $C;R$, since $S$ constrains only $x_1,x_2$. If interactions from $x_1,x_2$ change other qubits in a way that makes $L;C$ and $C;R$ inconsistent, the symbolic rule is not applicable. 
\end{example}

\subsection{Anchoring Canonical Rules}
\label{sec:anchoring}
Most canonical symbolic rules are size-preserving and are useful when we want to move $L$ and $R$ around the symbolic part $S$. The $L$ part can move to the right side of $S$ and become $R$. However, size-reducing rules can sometimes transform circuits directly, speeding up the search process, rather than first transforming the circuit using a size-preserving rule and then applying a size-reducing rule. We can do this by appending prefixes and suffixes to the canonical rule. However, naively enumerating all possible suffixes and prefixes will cause a combinatorial explosion, which is why we derive canonical rules instead of all possible rules in the first place. Canonical symbolic rules serve as the core, compact generative basis from which \sys can construct rules. They can be viewed as templates that can be instantiated into multiple concrete or symbolic transformations by prefixing or suffixing the circuits.
Yet, not all instantiations are meaningful. For example, left-appending or right-appending a circuit that implements the identity matrix $I$ to both sides produces a trivial symbolic rule.

To avoid generating redundant variants, \sys introduces a novel technique that selectively instantiates only those symbolic rules that are likely to "collaborate" with concrete rules to reduce gates by appending carefully chosen prefixes or suffixes to both sides.
\begin{figure}
\begin{mathpar}
\small
\inferrule[Init]
{ }
{\Delta := \Delta_{init}}

\inferrule[Append-1]
  {L;S;R = B_s \in \Delta \\ \mathit{suffix}\ L\ L_c \\ L_c \rightarrow R_c \in \Delta}
  {\Delta := \{L_c;s;R = \mathit{dropBack}(L_c,L);B_s\} \cup \Delta}
  \quad
\inferrule[Append-2]
  {A_s = L;S;R \in \Delta \\ \mathit{prefix}\ R\ L_c
  \\ L_c \rightarrow R_c \in \Delta}
  {\Delta := \{A_s;\mathit{dropFront}(L_c, R) = L; S; L_c\} \cup \Delta}

\inferrule[Append-3]
  {A_s = L;S;R \in \Delta \\ \mathit{suffix}\ L\ L_c \\ L_c \rightarrow R_c \in \Delta}
  {\Delta := \{\mathit{dropBack}(L_c, L);A_s = L_c; S; R\} \cup \Delta}
\quad
\inferrule[Append-4]
  {L;S;R = B_s \in \Delta \\ \mathit{prefix}\ R\ L_c \\ L_c \rightarrow R_c \in \Delta}
  {\Delta := \{L;S;L_c = B_s;\mathit{dropFront}(L_c,R)\} \cup \Delta}
\end{mathpar}
\caption{Inference rules for anchoring canonical rules}
\Description{Five inference rules define anchoring. The initialization rule starts from the generated rule set, and four append rules add prefixes or suffixes to canonical symbolic rules when those concrete fragments overlap a size-reducing concrete rule.}
\label{fig:anchor-rules}
\vspace{-1em}
\end{figure}
Fig.~\ref{fig:anchor-rules} shows the inference rules we use to derive anchored rules. $\Delta$ is the current rule context that stores all derived rules. $L,R,L_c,R_c$ are concrete subcircuits that contain no symbolic variables. $S$ is a symbolic subcircuit variable. $A_s$ and $B_s$ denote subcircuits that contain a symbolic variable. \textit{prefix/suffix} $X$ $Y$ indicates whether $X$ is a proper prefix/suffix of $Y$. \textit{dropfront/dropback} $X$ $Y$ indicates that, if $Y$ is a prefix/suffix of $X$, the operation drops $Y$ from the front/back of $X$ and yields the remainder. We assume that concrete rules $L_c \rightarrow R_c$ are size-decreasing. Rule \textsc{Init} initializes the rule set with all currently generated rules. \textsc{Append-2} first checks whether the concrete part $R$ of the rhs of the canonical symbolic rule is a prefix of the lhs $L_c$ of any concrete rule. If so, it can extend the rhs of the symbolic rule by transforming $R$ into $L_c$ and extend the lhs of the symbolic rule by adding, to the right of $A_s$, the part that remains after removing the prefix $R$ from $L_c$. The other \textsc{Append-x} rules are similar.
\begin{lemma}
   $\Delta$ always maintains a set of valid rules. 
\end{lemma}
\begin{lemma}
   If $R \in \Delta \setminus \Delta_{init}$ is applied to a circuit $C$ and results in $C'$, then a concrete rule in $\Delta$ can subsequently be applied to $C'$. We define the full set of anchored rules as $\Delta \setminus \Delta_{init}$ together with the canonical symbolic rules from which no additional rules are generated through anchoring.
\end{lemma}
\begin{example}
Consider the canonical symbolic rule from Example~\ref{ex:canonical}:
$rz(\gamma)\ q_0;\; S = S;\; rz(\gamma)\ q_1$. \sys first observes that the concrete suffix on the right-hand side, $rz(\gamma)\ q_1$, matches the prefix of the lhs of the concrete rule
$rz(\alpha)\ q_1;\; rz(\beta)\ q_1 \;\rightarrow\; rz(\alpha + \beta)\ q_1$. Note that matching the symbolic parameters creates a substitution $\delta = \{\alpha: \gamma\}$.
Hence, \sys can right-append $\delta(rz(\beta)\ q_1) = rz(\beta)\ q_1$ to both sides of the canonical symbolic rule, producing the new symbolic rule:
\begin{align}
rz(\gamma)\ q_0;\; S;\; rz(\beta)\ q_1 
&\;\rightarrow\; S;\; rz(\gamma)\ q_1;\; rz(\beta)\ q_1. \tag{5.1}
\end{align}
Although the new rule (5.1) only appends a suffix to both sides, it "enables" the application of a concrete rule in Fig.~\ref{fig:rewrite}(c), which can transform (5.1) to a size-reducing rule:
\begin{align}
rz(\gamma)\ q_0;\; S;\; rz(\beta)\ q_1 
&\;\rightarrow\; S;\; rz(\gamma + \beta)\ q_1. \tag{5.2}
\end{align}
Only the symbolic rule~(5.1) is retained as an anchored extension of the canonical rule, while the subsequent transformation~(5.2) can result from applying existing concrete rules during runtime optimization.
\end{example}
\section{Applying Rules}
We demonstrate that our generated rule sets can be easily integrated using standard techniques such as equality saturation~\cite{eqsat} and a search-based stochastic algorithm such as simulated annealing~\cite{simulated}. We also show in the next section that these simple techniques achieve good results.
Our generated rule sets are small and compact, allowing for the composition of a large number of rules. We easily integrate our concrete rules into existing equality-saturation engines~\cite{egglog,egg}. However, because of the limitations of equality saturation, an engine can run for only a small number of iterations before e-node explosion, when too many circuit terms are stored in the e-graph.
% The conventional matching algorithm matches subcircuits on a circuit graph by checking if the pattern and the subcircuits are isomorphic. We modeled quantum circuits as directly sequential lists of gates. The advantage of such representation is that the ASTs built from such representation are much simpler than circuit graphs and can be directly fed into the e-graph. One disadvantage is that it cannot represent parallel gates directly. However, we directly encode the commutation of parallel gates within a gate set into rewrite rules before optimizing circuits in that particular gate set. Thus, during optimization, the e-graph can run such commutation rules in a limited number of iterations to learn commutative gates within a limited window size, rather than learning all the commutative gates at once. 
Although the generated concrete rules can be easily used by equality saturation to compose larger rules, the e-graph is rarely able to saturate because it quickly fills with e-nodes when optimizing large circuits. In our experience, equality saturation can run for only up to 15 iterations before the number of e-nodes explodes. For a larger circuit, it can run for only around 5--10 iterations. Thus, we combine our concrete rules using equality saturation and apply symbolic rules using a simple stochastic algorithm that randomly selects and applies a symbolic rule. If the application reduces the current cost, it is accepted. Otherwise, it may be accepted with a small probability. Our intuition is that concrete rules can achieve global optimization within a limited circuit window (by running a limited number of equality-saturation iterations) and can quickly optimize short-distance subcircuits to an optimal size, whereas symbolic rules can achieve long-distance optimization by matching long-distance patterns and can serve as an exploration mechanism after concrete rules can no longer be applied. After an anchored rule is applied, concrete rules in an e-graph can again yield progress.
% In our evaluation, we set a window size of [10, 30] and show that it does not take long time for symbolic rules to match subcircuits and take effects to yield good optimization results.

Algorithm~\ref{alg:opt} shows the overall search-based optimization algorithm. The algorithm runs in a while loop until timeout. In each outer iteration, it runs 1 to $N$ iterations of equality saturation, alternating with the application of symbolic rules. In each symbolic step, it randomly selects a symbolic rule and matches it against the current search circuit. If a matched subcircuit is found, it applies the symbolic rule to obtain a new circuit. The \texttt{Match} method scans the circuit graph once per iteration (linear in gate count). $O(n^2)$ behavior arises only when many overlapping matches are attempted within the same window. However, with anchored symbolic rules, the added prefixes and suffixes limit a symbolic rule's matches to subcircuits that, after rewriting, enable a later size-reducing concrete rule. If the new circuit has a lower cost than the current circuit, it is accepted as the new search state. Otherwise, it may still be accepted with a small probability. It then builds an empty e-graph from the current circuit, saturates it with concrete rules for $k$ iterations, and extracts the lowest-cost circuit according to the cost model, which can be either the number of two-qubit gates or the total number of gates. After $N$ iterations, the algorithm starts over from $k = 1$. When timeout is reached, the best circuit encountered is returned. $T$ is a temperature hyperparameter that adjusts the acceptance probability~\cite{simulated}: larger $T$ yields higher acceptance probability.
We allow users to specify a window with a lower bound and an optional upper bound on the sizes of subcircuits that $S$ can match, such that a symbolic rule matches $S$ to a subcircuit larger than $l$ and smaller than $u$, while concrete rules find equivalences in subcircuits smaller than $l$. The lower bound ensures that symbolic and concrete rules do not overlap in their roles. This setting also provides flexible control: a larger or unlimited window allows $S$ to match larger subcircuits and prioritize broader transformations, while a smaller window focuses on smaller subcircuits so that each match runs faster and more rules can be explored. This does not mean that the system cannot optimize circuits larger than $u$: such optimizations can be performed through multiple symbolic rewrites.
\begin{figure}[t]
\centering
\resizebox{0.7\linewidth}{!}{
\begin{minipage}{0.8\textwidth}
\begin{algorithm}[H]
\small
\caption{Circuit Optimizer}
\label{alg:opt}
\begin{algorithmic}[1]
\Function{Optimize}{c, N, ConcreteRules, SymbolicRules, T}
\State $currentCircuit \gets c$
\State $bestCircuit \gets c$
\While {not Timeout}
\For {k in 1 to N}
\State $rule \gets$ RandomSelect(SymbolicRules)
\State $matchedSubcircuit \gets$ Match(currentCircuit, rule.lhs)
\If{$matchedSubcircuit \neq None$}
\State $newCircuit \gets$ ApplyRule(currentCircuit, rule, matchedSubcircuit)
\If{$Cost(newCircuit) < Cost(currentCircuit)$}
\State $currentCircuit \gets newCircuit$
\Else
\State $p \gets$ exp(-(Cost(newCircuit)-Cost(currentCircuit))/T)
\If{Random(0,1) < p}
\State $currentCircuit \gets newCircuit$
\EndIf
\EndIf
\EndIf
\If{$Cost(currentCircuit) < Cost(bestCircuit)$}
\State $bestCircuit \gets currentCircuit$
\EndIf
\State $eGraph \gets$ EmptyEGraph(currentCircuit)
\State $eGraph \gets$ eGraph.saturate(ConcreteRules, k)
\State $currentCircuit \gets$ ExtractBest(eGraph, CostModel)
\If{$Cost(currentCircuit) < Cost(bestCircuit)$}
\State $bestCircuit \gets currentCircuit$
\EndIf
\EndFor
\EndWhile
\State \textbf{return} bestCircuit
\EndFunction
\end{algorithmic}
\end{algorithm}
\end{minipage}
}
\Description{Pseudocode for circuit optimization. The algorithm repeatedly samples and applies symbolic rules using simulated-annealing acceptance, then builds an e-graph, saturates it with concrete rules for an increasing number of iterations, and extracts the lowest-cost circuit until timeout.}
\par
\end{figure}
\section{Implementation and Evaluation}
We implemented \sys in approximately 3.2k lines of Java and 1.0k lines of Python. The synthesis engine (including both concrete-rule and symbolic-rule synthesis) is implemented in Java. Our concrete rule term enumerator is built on Queso's enumerator~\cite{queso}. The constraint-checking module is implemented in Python using SymPy~\cite{sympy} to model and solve linear systems over symbolic matrices. For equality saturation, we use the open-source e-graph library egglog~\cite{egglog}.
In this section, we evaluate \sys and answer the following research questions.
\begin{itemize}
    \item (1) How do \sys's rules compare to rewrite rules generated by other quantum synthesis engines?
    \item (2) How do the formalized techniques help speed up the process of symbolic rule generation?
    \item (3) How does \sys, when used with simple optimization techniques, compare to existing state-of-the-art quantum circuit optimizers?
    \item (4) What effect do symbolic rules have on optimization? What is the difference among using only concrete rules, using canonical symbolic rules, and using anchored rules?
\end{itemize}
\paragraph{\textbf{Experimental Setup}} 
All experiments were run on an AMD Ryzen 9 9950X (16 cores) with 128GB RAM, running Ubuntu 22.04.
\subsection{Rewrite Rule Synthesis}
\begin{table}[h]
\centering
\small  % or \footnotesize / \scriptsize
\caption{Concrete rule-set size (number of rules) generated by \sys that fully derive prior rules}
\vspace{-1em}
\label{tab:size}
\begin{tabular}{r|r|r|r|r|r}
Gateset & Gates & \#Rules (\sys) & \#Rules (Prior) & $n$ & Time \\
\hline
ibm-eagle &\makecell{cx, rz \\ x, sx}& \textbf{179} & 4615 (Queso) & 5 & 770s\\
\hline
nam & \makecell{x h \\ rz cx}& \textbf{194} & 8816 (Quartz, Queso) & 5 & 341s\\
\hline
rigetti & \makecell{rx1 rx2 \\ rx3 rz cz} & \textbf{46} & 8773 (Quartz, Queso) & 5 & 484s\\
\hline
ion & \makecell{rx ry \\ rz rxx} &\textbf{180} & 11777 (Queso) & 3 & 237s\\
\hline
\end{tabular}
\vspace{-1em}
\end{table}
\paragraph{\textbf{Size Comparison.}}
We evaluate the concrete rules generated by \sys on four gate sets: (1) IBM-Eagle for IBM devices, (2) Rigetti's gate set, (3) an ion-trap gate set~\cite{ionq2022native}, and (4) the gate set of Nam et al.~\cite{nam2018automated}. The IBM and Rigetti gate sets target superconducting devices, which currently constitute the largest physically realized quantum hardware.
We conduct a size-comparison study to show that our concrete rule set (rule-set size = number of rules) remains compact while still deriving rules from prior work. We collect the gate-set grammars used by previous systems~\cite{queso, quartz} so that the generated \sys rules can cover their concrete rule spaces for a fair comparison. Table~\ref{tab:size} reports both rule-set size (number of rules) and the maximum rule-length bound used during synthesis. For the IBM-Eagle, Nam, and Rigetti gate sets, we use a maximum rule length of 5, whereas for the ion gate set, we use a rule length of 3. Quartz does not report rules for the IBM-Eagle and ion gate sets. We choose synthesis bounds that allow synthesis to finish within 15 minutes and show optimization effectiveness in the next subsection.

% Queso uses pure probabistic grouping without SMT equivalence checking, which can result false positive in low probability. QSymb uses Queso's method for grouping and validate the final rules in $L$ using SMT checking and constains an additional filtering process using an e-graph. Quartz enumerates all the rules within $(n,q)$ and use SMT checking.

\paragraph{Results.}
Across all four gate sets, \sys generates concrete rule sets that are 26x--191x smaller while still deriving all rules generated by Queso and Quartz.
% {\textbf{Results.}}~For all of 4 gate sets, \sys generates a much smaller concrete rule set and is able to derive all the rules generated by Queso and Quartz.

\paragraph{\textbf{Symbolic Rules.}}
\begin{wraptable}{r}{0.45\textwidth}
\vspace{-1.5em}
    \centering
    \caption{Anchoring cost of symbolic rules}
    \vspace{-1em}
    \begin{tabular}{|r|r|r|r|}
        \textbf{Gateset} & \textbf{\#Rules} & \textbf{Time (s)} & \textbf{size}\\
         ibm-eagle & 436 & 0.6 &5 \\
         nam & 484 & 1.1 & 5\\
         rigetti & 97 & 0.5 & 5\\
         ion & 607 & 0.8 & 3\\
    \end{tabular}
    \vspace{-1em}
    \label{tab:anchor}
\end{wraptable}
Because \sys synthesizes symbolic rules that differ from prior approaches, we compare end-to-end optimization performance against tools that use monomial symbolic rules in the next subsection.
Here, we evaluate how our formalized techniques accelerate symbolic-rule synthesis and filter large numbers of invalid candidates.
To measure the effect of property-based grouping, we compare \sys with grouping enabled and disabled. Since a symbolic gate can match arbitrarily long subcircuits, in practice we generate only canonical symbolic rules with short concrete parts, then use anchoring to construct longer, optimization-effective rules. 
In practice, we use a canonical symbolic rule size of 3 for the IBM, Nam, and Rigetti gate sets, and 2 for the ion gate set; all are synthesized within 5 minutes. Table~\ref{tab:grouping} reports the number of canonical symbolic rules and synthesis time with and without property-based grouping. We then anchor canonical rules to generate longer symbolic rules. Table~\ref{tab:anchor} reports the number of anchored symbolic rules and anchoring time.
\begin{table}[h]
    \centering
    \caption{Synthesis of canonical symbolic rules with and without property grouping}
    \vspace{-1em}
    \label{tab:grouping}
    \begin{tabular}{|r|r|r|r|r|r}
        Gateset & \#Rules & Cost w/ Grouping (s) & Cost w/o Grouping (s) & Speedup & size\\
         ibm-eagle & 101 & \textbf{2.4} & 40.9 & \textbf{17.0x} &3\\
         nam & 139 & \textbf{8.4} & 61.0 & \textbf{7.3x} &3\\
         rigetti & 81 & \textbf{19.0} & 199.5 & \textbf{10.5x} & 3\\
         ion & 59 & \textbf{2.4} & 27.5 & \textbf{11.5x} &2 \\
    \end{tabular}
    \vspace{-1em}
\end{table}
% \textbf{Results}. Grouping candidates by properties have significantly reduce the search space and achieve x speedup.
% \paragraph{Compare to General Rewrite Inference Tool}
% To further show the performance of \sys, we compare \sys with a general rewrite inference tool\cite{} that can generate nonderivable rewrite rules for any domain given its semantics and a grammar. We use the same grammar and semantics as \sys to generate rewrite rules using the general rewrite inference tool. We then compare the performance of the generated rules by both tools using equality saturation under the same grammar and semantics. Because it only generates concrete rules, we only compare the performance of concrete rule generation.

\paragraph{Results.}
Property-based grouping significantly reduces the search space for canonical symbolic-rule generation, yielding an average speedup of 11.6x. Without grouping, synthesis is inefficient because candidates from both sides of a rule must be enumerated and checked one by one. Table~\ref{tab:anchor} shows that anchoring runs in seconds, because we anchor only canonical rules that can enable a later concrete rewrite and we generate a small number of concrete rules.
\paragraph{\textbf{Summary of Questions 1 and 2.}}
The concrete rewrite rules generated by \sys are much smaller than those of prior quantum rewrite synthesizers, while still semantically covering their rules. Grouping by symbolic properties significantly accelerates symbolic-rule synthesis. Anchoring efficiently derives useful rules from the canonical rule set without exhaustive enumeration.
\subsection{Optimization Performance}
To evaluate symbolic-rule quality and application effectiveness, we conduct three experiments: (1) comparing full \sys (concrete + symbolic rules) against state-of-the-art optimizers; (2) comparing full \sys against two ablations (concrete-only and concrete + canonical symbolic rules); and (3) measuring the effect of anchoring on optimization performance.
% \paragraph{\textbf{Expectation}}
% \begin{itemize}
%     \item (1) We expect \sys to use symbolic rules and concrete rules to perform as good as or outperform other tools that use far more symbolic rules and concrete rules 
%     \item (2) We expect \sys that using symbolic rules and concrete rules perform better than using concrete rules.
% \end{itemize}

\paragraph{\textbf{Metrics}} We use two-qubit-gate reduction to evaluate \sys. We focus on two-qubit gates because, on NISQ hardware, they typically have much higher error rates than single-qubit gates. The metric is the percentage reduction in two-qubit-gate count relative to the original circuit. 
$$ \text{Percentage Reduction} = \frac{\text{Original Count} - \text{Optimized Count}}{\text{Original Count}} * 100\% $$
In addition, to evaluate practical quality, we compute circuit fidelity. The fidelity of a gate $g$ is $1 - error\_rate(g)$, and the fidelity of a circuit is $\prod_{i=1}^n fidelity(g_i)$. Gate error rates are taken from Qiskit's calibration data for IBM Washington~\cite{qiskit}.
% (1) We first evaluates how \sys performs in optimizing quantum circuits by only using concrete rules versus using both symbolic rules and concrete rules under same benchmarks. (2) We then counts the number of new long distance rules symbolic rules has discovered after equality saturation is run for concrete rules, meaning the number of equivalences that equality saturation hasn't found after n iterations. 
\paragraph{\textbf{Benchmarks}} We use benchmarks from prior rewrite-based optimizers~\cite{queso, quartz} and approximate optimization work~\cite{quest}. Our benchmark suite contains 135 circuits spanning near-term and long-term algorithms, including QAOA~\cite{farhi2014quantumapproximateoptimizationalgorithm}, VQE~\cite{Peruzzo_2014}, QPE~\cite{kitaev1995quantummeasurementsabelianstabilizer}, QFT~\cite{coppersmith2002approximate}, Grover~\cite{grover}, and Shor~\cite{shor}. It also includes benchmark families used by Queso and Quartz. All input circuits are decomposed into the target gate set. The circuits contain an average of 1371.8 gates; the maximum is 17438, and the minimum is 15. They use an average of 13.3 qubits; the maximum is 36, and the minimum is 4.
\paragraph{\textbf{Experiment Setup}} We set temperature $T=10.0$ and run $N=9$ equality-saturation iterations in Algorithm~\ref{alg:opt}. Each target has a 60-minute timeout. We evaluate on the IBM-Eagle and Nam gate sets, using all generated concrete rules and anchored symbolic rules, with 3 trials per benchmark. We set the symbolic-match window to $[10, \infty]$ so that symbolic and concrete rules do not overlap in the same local region. Each data point is averaged over 3 trials. We compare against state-of-the-art rewrite-based tools: Quartz~\cite{quartz}, Queso~\cite{queso}, Qiskit~\cite{qiskit}, and TKET~\cite{tket}. We also compare with Guoq~\cite{guoq}, which combines resynthesis and rewriting; for fairness, we compare against its rewrite engine. Guoq's rewrite engine is adapted from Queso and therefore uses a similar rule set. Among the compared tools, Queso and Guoq use monomial symbolic rules, Quartz uses concrete rules with symbolic parameters, and Qiskit/TKET use only concrete rules. We run Guoq, Quartz, and Queso using the settings and trial procedures reported in their papers. For Qiskit, we use optimization level 3. Table~\ref{tab:tools} summarizes each tool's approach and symbolic-rule support. Column "Used Symb" denotes whether the tool uses symbolic rules. Column "Synthesize Symb" denotes whether the tool itself synthesizes symbolic rules.
\paragraph{\textbf{Results.}}

\begin{figure}[h]
    \includegraphics[width=\textwidth]{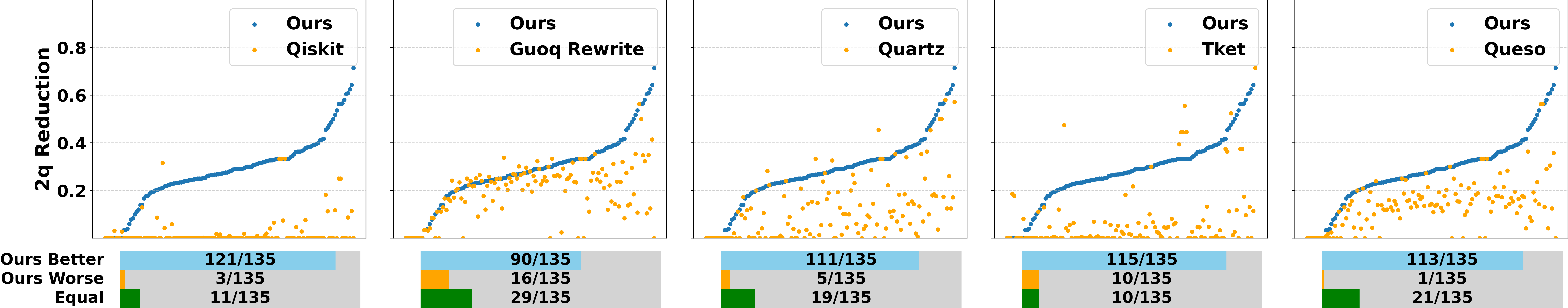}
    \includegraphics[width=\textwidth]{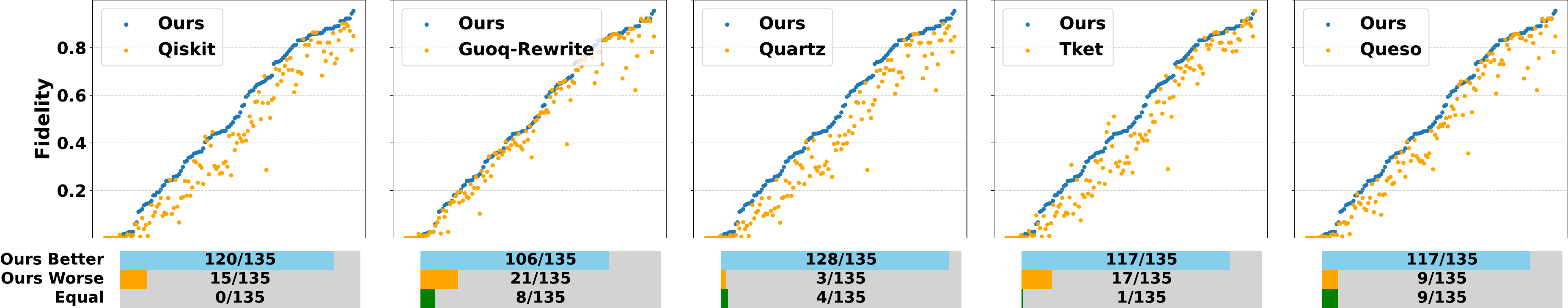}
    \caption{\sys versus state-of-the-art optimizers on IBM-Eagle gateset}
    \Description{Two benchmark-ordered comparison plots for the IBM-Eagle gate set. The first compares two-qubit-gate reduction and the second compares estimated circuit fidelity across \sys, Qiskit, Guoq, Quartz, TKET, and Queso. \sys achieves the strongest result on most benchmarks.}
    \label{fig:vs_sota}
    \vspace{-1em}
\end{figure}

\begin{figure}[h]
    \includegraphics[width=\textwidth]{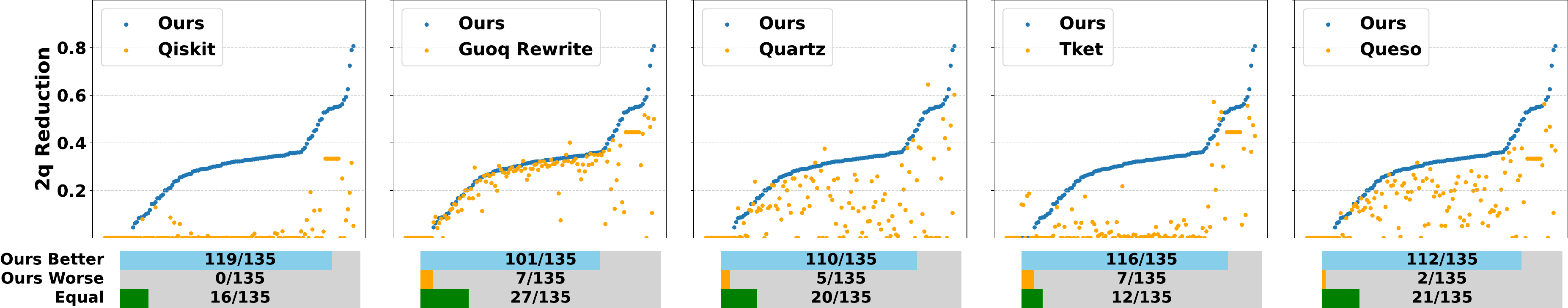}
    \caption{\sys versus state-of-the-art optimizers on Nam gateset}
    \Description{A benchmark-ordered plot comparing two-qubit-gate reduction on the Nam gate set across \sys, Qiskit, Guoq, Quartz, TKET, and Queso. \sys achieves the strongest reduction on most benchmarks.}
    \label{fig:vs_sota_nam}
    \vspace{-1em}
\end{figure}

\begin{table}[h]
    \centering
    \small
    \caption{Comparison of tools and their approaches}
    \vspace{-1em}
    \label{tab:tools}
    \begin{tabular}{c|c|c|c}
         \text{Tool} & \text{Opt Approach} & \text{Used Symb}  & \text{Synthesize Symb}\\
        \hline
         \sys & Eqsat \& Simulated Annealing & $\checkmark$ & full infinite space\\
         Guoq-Rewrite~\cite{guoq} & Simulated Annealing &  $\checkmark$ & $\times$\\
         Queso~\cite{queso} & Beam Search & $\checkmark$ & monomial finite space\\
         TKET~\cite{tket} & fixed passes & $\times$ & $\times$\\
         Qiskit~\cite{qiskit} & fixed passes & $\times$ & $\times$\\
         Quartz~\cite{quartz} & Beam Search & $\times$ & $\times$\\
    \end{tabular}
    \vspace{-1em}
\end{table}
Fig.~\ref{fig:vs_sota} shows two-qubit reduction and fidelity against state-of-the-art optimizers on the IBM-Eagle gate set. The x-axis orders benchmarks by \sys's result (from smaller to larger reduction/fidelity). We observe that \sys outperforms other optimizers on most benchmarks. In two-qubit reduction, \sys strictly outperforms Qiskit~\cite{qiskit}, Guoq~\cite{guoq}, Quartz~\cite{quartz}, TKET~\cite{tket}, and Queso~\cite{queso} on 90\%, 67\%, 82\%, 85\%, and 83\% of benchmarks, respectively. For fidelity, \sys strictly outperforms Qiskit~\cite{qiskit}, Guoq~\cite{guoq}, Quartz~\cite{quartz}, TKET~\cite{tket}, and Queso~\cite{queso} on 89\%, 78\%, 94\%, 87\%, and 87\% of the benchmarks. Search-based tools (Quartz, Queso, and Guoq) continue optimization until timeout, while fixed-pass tools (Qiskit and TKET) terminate after their pass sequence completes. The average runtimes for Qiskit and TKET are 9.3s and 12.1s, respectively. Fig.~\ref{fig:vs_sota_nam} shows two-qubit reductions for the Nam gate set across optimizers. Since Nam does not have real-hardware fidelity data, we show only two-qubit reduction. \sys strictly outperforms Qiskit~\cite{qiskit}, Guoq~\cite{guoq}, Quartz~\cite{quartz}, TKET~\cite{tket}, and Queso~\cite{queso} on 88\%, 74\%, 81\%, 86\%, and 82.9\% of the benchmarks, respectively.

\begin{wrapfigure}{r}{0.3\textwidth}
    \centering
    \begin{minipage}{0.29\textwidth}
    \includegraphics[width=\textwidth]{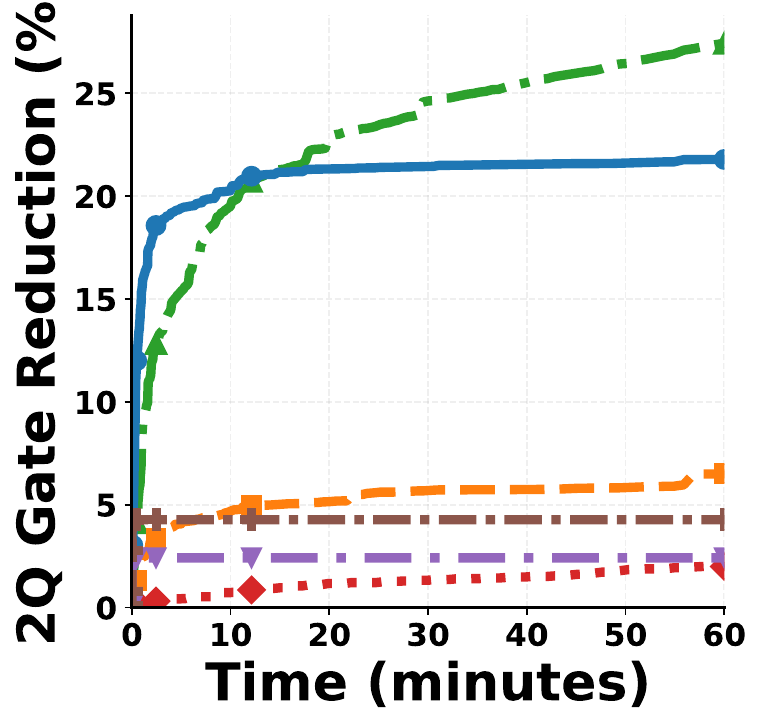}
    \label{fig:runtime}
    \end{minipage}
    \begin{minipage}{0.29\textwidth}
    \vspace{-1.6em}
    \includegraphics[width=\textwidth]{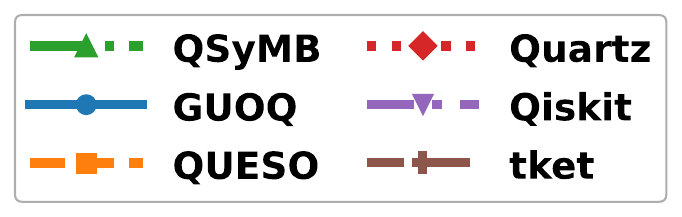}
    \label{fig:runtime_legend}
    \end{minipage}
    \Description{A line plot of average two-qubit-gate reduction over optimization time on IBM-Eagle benchmarks, with a separate legend. Fixed-pass tools stop early, while \sys continues improving and reaches the highest final reduction among the compared optimizers.}
\vspace{-1.5em}
\end{wrapfigure}
\textbf{Runtime.} The figure on the right shows a runtime comparison of the average two-qubit reduction against state-of-the-art optimizers on the IBM-Eagle gate set, computed as $1 - \frac{\mathit{mean}(\sum \mathrm{optimized\_2q})}{\mathit{mean}(\sum \mathrm{original\_2q})}$. The results for the Nam gate set are similar to those for IBM-Eagle, so we omit them for brevity.
As shown, \sys starts slower than Guoq in the first 10 minutes. However, as Guoq slows down, \sys continues to make progress, converges faster than Guoq, and reaches better final results. For fixed-sequence optimizers such as Qiskit and TKET, optimization terminates after a few seconds and the results then remain unchanged.
\sys performs significantly better than Queso and Quartz, both by converging faster and by achieving better final results. \sys reaches a final average two-qubit reduction rate of 27.44\% and average fidelity of 0.49 on IBM-Eagle; \sys reaches a final average two-qubit reduction rate of 29.95\% on the Nam gate set.
% Fig.~\ref{fig:vs_time} also records the average 2-qubit reduction rate over time for search-based optimizers (\sys, Guoq, Queso, and Quartz) over 10 minutes, since most optimization progress happens in the first 10 minutes and all tools change very little afterward. We measure 2-qubit reduction rate per minute. We observe that all tools begin to optimize more slowly after 5 minutes, but \sys maintains higher overall optimization quality.
\paragraph{Summary of Question 3}
By combining compact concrete rules, expressive symbolic rules for long-distance matching, and anchoring, \sys significantly outperforms state-of-the-art rewrite-rule-based optimizers in both convergence rate and final results. The runtime results show that \sys continues to make progress after ten minutes, in contrast to prior work, demonstrating the effectiveness of our synthesized optimizations.

\subsection{Ablation Study}
We further evaluate the impact of symbolic rules by comparing full \sys against two ablations: \sys with only concrete rules, and \sys with concrete rules plus canonical symbolic rules. To isolate the effect of our concrete-rule synthesis algorithm, we also compare the concrete rules generated by Algorithm~\ref{alg:concrete} against the same number of rules randomly selected from more than six thousand rules generated by Queso. We use the same temperature and iteration settings as in the previous experiment. Additional ablation studies are provided in Appendix C, including the effect of symbolic-rule window sizes and the impact of anchored rules versus canonical rules over time.

\begin{figure}[h]
    \centering
    % \begin{minipage}{0.38\textwidth}
    % \centering
    % \includegraphics[width=\textwidth]{figs/quality_vs_time.png}
    % \caption{\sys vs. search-based optimizers over 10 minutes}
    % \label{fig:vs_time}
    % \end{minipage}
    \begin{minipage}{0.66\textwidth}
    \includegraphics[width=\textwidth]{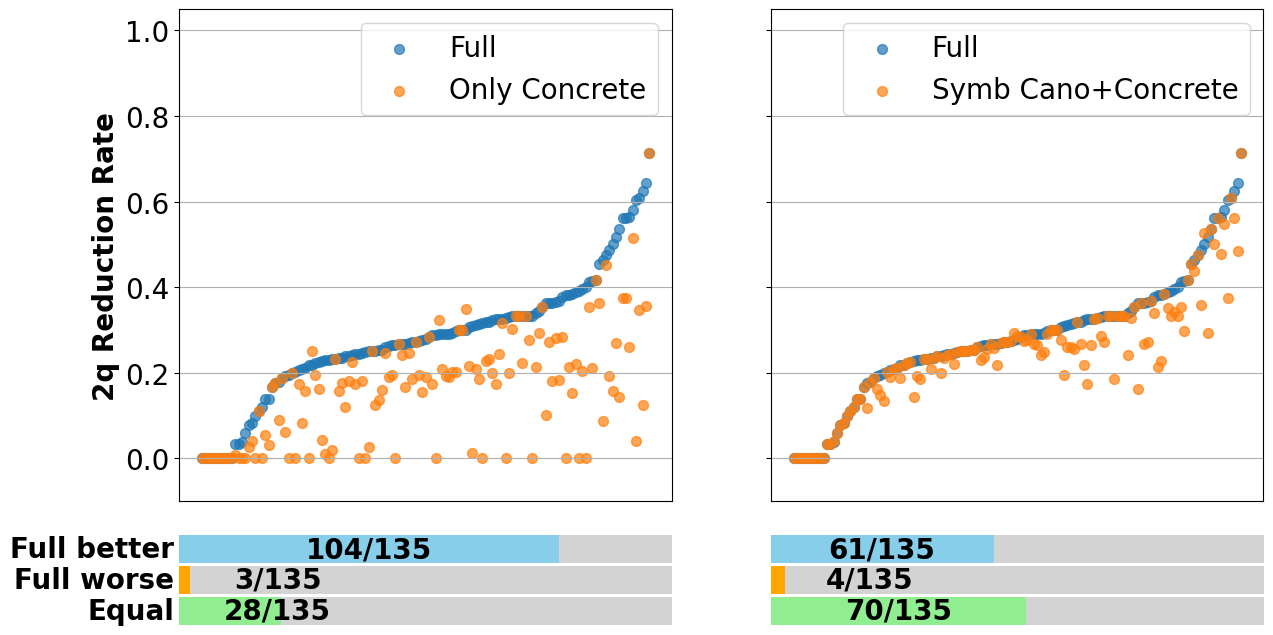}
    \caption{Full \sys vs. \sys with only concrete rules vs. \sys with canonical symbolic rules}
    \Description{Ablation plots compare full \sys with a concrete-only variant and with a variant using concrete plus canonical symbolic rules. Full \sys generally produces greater two-qubit-gate reduction, showing the benefit of symbolic rules and anchoring.}
    \label{fig:concrete vs symb}
    \end{minipage}
    \vspace{-1.2em}
\end{figure}
\textbf{Results.} Fig.~\ref{fig:concrete vs symb} compares full \sys against two ablations: \sys with only concrete rules, and \sys with concrete rules plus only canonical symbolic rules. The results show that full \sys performs as well as or better than the concrete-only variant in 97\% of cases and is strictly better in 78\% of cases. With anchored symbolic rules---formed by appending prefixes and suffixes to many canonical rules---full \sys further improves optimization, performing as well as or better than the canonical-symbolic-plus-concrete variant in 97\% of cases, and strictly better in 31\% of cases. In the best case, full \sys improves two-qubit-gate reduction by 21\%.
Fig.~\ref{fig:concrete vs random} compares \sys using the concrete rules generated by Algorithm~\ref{alg:concrete} against an equal number of rules randomly selected from the Queso rule set~\cite{queso} of over six thousand rules.
For a fair comparison, we perform three trials. In each trial, we independently sample a new random rule set with the same number of rules. Fig.~\ref{fig:concrete vs random} reports the mean across the three trials. The randomly selected rules make almost no optimization progress because each set consists of 179 rules randomly selected from more than six thousand and is unlikely to derive the complete equivalences within the $(n,q)$ bounds.
\begin{wrapfigure}{r}{0.4\textwidth}
    \centering
    \begin{minipage}{0.38\textwidth}
    \includegraphics[width=\textwidth]{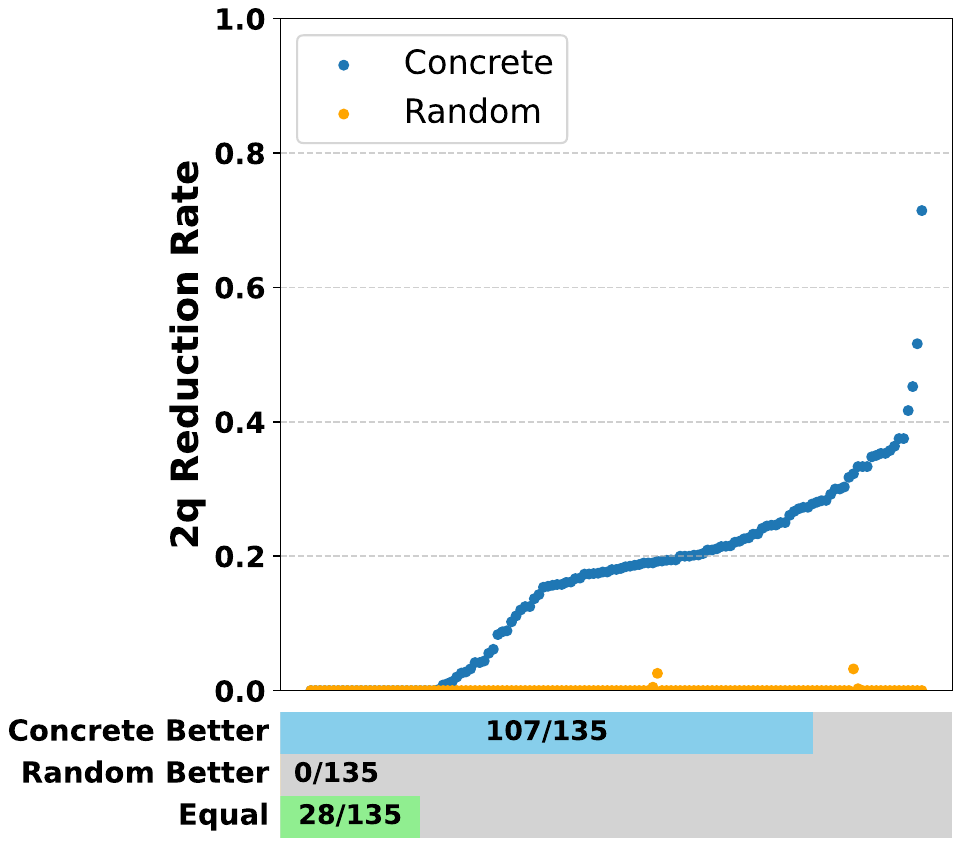}
    \caption{Concrete \sys vs. the same number of randomly selected Queso rules}
    \Description{A benchmark-ordered comparison of two-qubit-gate reduction using \sys's compact synthesized concrete rule set versus an equally sized random subset of Queso rules. The synthesized set consistently makes substantially more optimization progress.}
    \label{fig:concrete vs random}
    \end{minipage}
    \vspace{-1.5em}
\end{wrapfigure}
\paragraph{Summary of Question 4} The results show that the combined approach achieves significantly better optimization than either component alone. Rule anchoring further improves optimization beyond what is achieved by canonical symbolic rules and concrete rules alone. The concrete-vs-random comparison also shows the importance of Algorithm~\ref{alg:concrete}: a compact, non-derivable, and $(n,q)$-complete concrete rule set is substantially more effective than an equally sized random subset of a much larger rule set.
\section{Discussion and Future Work}
As in prior bounded rewrite-rule synthesizers~\cite{quartz, queso}, the search space grows exponentially with qubit count and rule size. In practice, however, we choose rule sizes that allow optimization to converge quickly while still producing strong results. Thus, like Queso~\cite{queso} and Quartz~\cite{quartz}, \sys restricts synthesis to rules over at most 3 qubits. For example, Quartz reports that larger rule sizes ($n$) or rule qubit counts ($q$) do not necessarily yield better optimization results and can sometimes degrade performance. For \sys, most larger-size concrete rules are eliminated by the e-graph, so searching for larger-size rules that will mostly be eliminated is not an effective strategy. Our evaluation shows that simple optimizers using small-size non-derivable concrete rules together with e-graphs and symbolic rules can outperform existing work.

One interesting direction for advancing the optimization algorithm is to use machine-learning techniques to guide the application of symbolic rules. For example, using an LLM to analyze the structure of a circuit and apply symbolic rules in a guided way could potentially improve optimization performance. Additionally, integrating \sys's rewrite-rule engine with resynthesis techniques~\cite{quest, doecode_58510} to further improve optimization is left for future work.
\section{Related Work}
\paragraph{Synthesis of Quantum Optimizations}
Several works have explored the synthesis of quantum circuit optimizations. Queso~\cite{queso} is a quantum circuit optimization rule synthesizer that generates rewrite rules and checks equivalence using efficient polynomial testing, but it does not perform SMT-based checking, which can lead to low-probability false positives. Queso only supports symbolic circuits with monomial gates, which represent only a small and finite fraction of the solution space captured by symbolic rules. Quartz~\cite{quartz} is another quantum circuit optimization rule synthesizer that generates rewrite rules by enumerating all possible circuits within given size and qubit bounds, and it supports symbolic parameters such as angles. \sys follows Queso's approach to group circuits and uses equality saturation for concrete-rule synthesis, resulting in a much smaller rule set that is validated by an SMT solver while ensuring $(n,q)$-completeness. In addition, unlike these two prior works, \sys supports synthesis of symbolic rules whose symbolic variables represent infinite interpretations of subcircuits, and introduces a property-grouping method that enables the generation of such expressive symbolic rules.

\paragraph{Quantum Circuit Optimizers}
Most quantum-circuit compilers use a fixed set of hand-crafted optimizations~\cite{qiskit,AmyGheorghiu2020_staq,CampbellEtAl2023_Superstaq,VOQC, nam2018automated, tket,Certiq}. PyZX~\cite{KissingerVandeWetering2020_PyZX} represents a circuit as a ZX-diagram and applies the graphical rewrite rules of ZX calculus.
Many works have developed verified compilers for quantum circuits; for example, VOQC~\cite{VOQC} is a formally verified quantum circuit optimizer built on top of Coq. Giallar~\cite{Giallar} utilized an automated verification approach based on SMT solvers to verify Qiskit optimization passes. Quarl~\cite{quarl} utilizes reinforcement learning to schedule the application of rules generated by Quartz. Some other works have used resynthesis techniques to optimize quantum circuits~\cite{doecode_58510, quest}. The approach partitions circuits into several subcircuits and resynthesizes each subcircuit's unitary matrix to get a new optimized subcircuit. Many other optimization works~\cite{churchill2017sound, taso, souper_pruning, minotaur} adopt superoptimization, finding the optimal solution for small programs that are later used in peephole optimizers. \sys utilizes a symbolic unitary matrix to represent the constraint of a symbolic subcircuit in a rewrite rule.

\paragraph{Equality Saturation}
With recent progress in fast, efficient equality-saturation engines~\cite{egglog, egg}, many works have used equality saturation~\cite{eqsat} to optimize programs in various domains, including classical programs~\cite{egglog, egg}, math libraries~\cite{herbie}, and tensor programs~\cite{tensat}. \sys's concrete rule inference is most related to general equality-saturation-based rule-inference engines such as Ruler~\cite{ruler} and its successor Enumo~\cite{enumo}. \sys integrates a domain-specific polynomial identity filter~\cite{queso} to group quantum circuits and employs circuit-pruning techniques to reduce the search space. Quasar~\cite{quasar} is concurrent work that runs sequential and graph-level equality saturation in parallel on quantum circuits and uses equality saturation to filter concrete rules from prior work. \sys uses equality saturation to synthesize concrete rules bottom-up for any given gate set $G$ and synthesizes symbolic rewrite rules.
% One important improvement over Ruler~\cite{ruler} is that the e-graph does not add all terms at once; instead, \sys uses efficient probabilistic equivalence checking for quantum circuits to group terms into equivalence classes before e-graph construction. Because there are too many quantum circuits even at gate size 6, adding all enumerated terms to the e-graph at once would create too many e-nodes and e-classes, even after merging equal terms derived by learned rules. For example, there are more than 10k e-classes after pruning when we enumerate terms within 3 qubits and size 6 for the IBM gate set. We realized that we do not need e-graphs for term grouping and instead initialize a new e-graph for each equivalence class. After grouping, each e-graph only needs to add terms from one equivalence class at a time, which reduces construction overhead. Existing general approaches add all terms to the graph at once, leaving a large number of e-classes even after merging, which does not scale to an enormous number of quantum circuits.

\section{Acknowledgments}
We thank the anonymous reviewers for valuable feedback that helped improve
this paper. This work was supported in part by an NSF CAREER Award CCF-2239484, an Amazon Research Award, a VMware Systems Research Award, and an NSF grant CCF-2124080. Ronghui Gu is a co-founder
of and has an equity interest in CertiK.

\section{Data-Availability Statement}
The source code, benchmarks, and scripts are open-sourced at \url{https://github.com/VeriGu/QSymb}. The artifact is also archived on Zenodo~\cite{qsymb-artifact}.

\bibliographystyle{ACM-Reference-Format}
\bibliography{ref}
\end{document}

%% file: macro.tex
\newcommand{\sys}{\textsc{Qsymb}\xspace}

\newcommand{\weiqiang}[1]  {\textcolor{orange}{(WQ: #1)}}
\newcommand{\ganxiang}[1]  {\textcolor{violet}{(GX: #1)}}
\definecolor{S}{HTML}{66CC00}